\documentclass[a4paper,11pt]{article}
\pdfoutput=1
\usepackage{a4wide}
\usepackage[UKenglish]{babel}
\usepackage[noadjust]{cite}
\usepackage{amsmath}
\usepackage{amssymb}
\usepackage{amsthm} 
\usepackage{bm}
\usepackage{bbm}
\usepackage{mathtools}
\usepackage{mathrsfs}
\usepackage{csquotes}
\usepackage{todonotes}
\usepackage[title]{appendix}
\usepackage{bbm}
\usepackage{xcolor}
\usepackage{blindtext}
\usepackage{enumerate}
\usepackage{manfnt}
\usepackage{multicol}
\usepackage{microtype}

\MakeRobust{\ref}% avoid expanding it when in a textual label

\makeatletter
\newcommand{\labeltext}[2]{%
  \@bsphack
  \csname phantomsection\endcsname % in case hyperref is used
  \def\@currentlabel{#1}{\label{#2}}%
  \@esphack
}
\makeatother

\usepackage[noframe]{showframe}
\usepackage{framed}
 {\endMakeFramed}
\definecolor{shadecolor}{gray}{0.88}

\usepackage{caption}
\usepackage{subcaption}
\usepackage[export]{adjustbox}
\usepackage{hyperref} 
\hypersetup{colorlinks=true,citecolor=blue,linkcolor=blue,urlcolor=blue}

\newcommand{\NP}{\operatorname{NP}}
\newcommand{\YES}{\textsc{Yes}}
\newcommand{\NO}{\textsc{No}}
\newcommand{\Test}[2]{
\def\temp{#2}\ifx\temp\empty
  \operatorname{Test}_{#1}
\else
  \operatorname{Test}_{#1}^{#2}
\fi
}

\newcommand{\bin}[2]{
{{#1}\choose{#2}}
}

\newcommand{\A}{\mathbf{A}}
\newcommand{\B}{\mathbf{B}}
\newcommand{\C}{\mathbf{C}}
\newcommand{\GG}{\mathbf{G}}
\newcommand{\HH}{\mathbf{H}}
\newcommand{\JJ}{\mathbf{J}}
\newcommand{\K}{\mathbf{K}}

\newcommand{\X}{\mathbf{X}}

\newcommand{\Vset}{V}
\newcommand{\Eset}{E}

\newcommand{\N}{\mathbb{N}}
\newcommand{\R}{\mathbb{R}}
\newcommand{\Q}{\mathbb{Q}}
\newcommand{\Z}{\mathbb{Z}}

\renewcommand{\vec}[1]{\mathbf{#1}}
\newcommand{\ba}{\vec{a}}

\newcommand{\bh}{\vec{h}}

\newcommand{\bx}{\vec{x}}
\newcommand{\bv}{\vec{v}}
\newcommand{\by}{\vec{y}}
\newcommand{\bw}{\vec{w}}
\newcommand{\be}{\vec{e}}

\newcommand{\bz}{\vec{z}}

\newcommand{\blambda}{{\bm{\lambda}}}
\newcommand{\bmu}{{\bm{\mu}}}

\DeclareMathOperator{\Tr}{Tr}

\DeclareMathOperator{\spann}{span}
\DeclareMathOperator{\dist}{dist}

\DeclareMathOperator{\diag}{diag}

\DeclareMathOperator{\Aut}{Aut}

\DeclareMathOperator{\Orb}{\mathscr{O}}

\DeclareMathOperator{\SDA}{SDA}
\DeclareMathOperator{\BLP}{BLP}
\DeclareMathOperator{\AIP}{AIP}
\DeclareMathOperator{\BA}{BA}
\DeclareMathOperator{\SDP}{SDP}
\DeclareMathOperator{\PCSP}{PCSP}
\DeclareMathOperator{\CSP}{CSP}

\DeclareMathOperator{\supp}{supp}

\DeclareMathOperator{\adj}{\mathscr{A}}
\DeclareMathOperator{\Alg}{Alg}

\newcommand{\bone}{\mathbf{1}}  
\newcommand{\bzero}{\mathbf{0}} 
\newcommand\Frob[2]{{\ensuremath\langle #1,#2\rangle}_{\operatorname{F}}}

\theoremstyle{plain}
\newtheorem{thm}{Theorem}
\newtheorem*{thm*}{Theorem}
\newtheorem{lem}[thm]{Lemma}
\newtheorem*{lem*}{Lemma}
\newtheorem{prop}[thm]{Proposition}
\newtheorem*{prop*}{Proposition}
\newtheorem{cor}[thm]{Corollary}
\newtheorem*{cor*}{Corollary}

\newtheorem*{conj*}{Conjecture}
\newtheorem{fact}[thm]{Fact}
\theoremstyle{definition}
\newtheorem{defn}[thm]{Definition}
\newtheorem*{defn*}{Definition}
\newtheorem{rem}[thm]{Remark}

\newcommand{\propparteq}[2]
{
\ensuremath{\stackrel{\operatorname{P}.\ref{#1}#2}{\;\;=\;\;}}
}
\newcommand{\equationeq}[1]
{
\ensuremath{\stackrel{\eqref{#1}}{\;\;=\;\;}}
}
\newcommand{\spaceeq}
{
\ensuremath{\stackrel{}{\;\;=\;\;}}
}

\begin{document}

\author{Lorenzo Ciardo\\
University of Oxford\\
\texttt{lorenzo.ciardo@cs.ox.ac.uk}
\and
Stanislav \v{Z}ivn\'y\\
University of Oxford\\
\texttt{standa.zivny@cs.ox.ac.uk}
}

\title{Semidefinite programming and linear equations \\ vs.\\ homomorphism problems\thanks{An extended abstract of this work appeared in the Proceedings of the 
56th Annual ACM Symposium on Theory of Computing (STOC'24).
This work was supported by UKRI EP/X024431/1. For the purpose of Open Access, the authors have applied a CC BY public copyright licence to any Author Accepted Manuscript version arising from this submission. All data is provided in full in the results section of this paper.}}

\date{\today}
\maketitle

\begin{abstract}
   \noindent We introduce a relaxation for homomorphism problems that combines semidefinite programming with linear Diophantine equations, and propose a framework for the analysis of its power based on the spectral theory of association schemes. We use this framework to establish an unconditional lower bound against the 
   semidefinite programming + linear
   equations model,
   by showing that the relaxation does not solve the approximate graph homomorphism problem and thus, in particular,
   the approximate graph colouring problem.
\end{abstract}

\section{Introduction}

\noindent Semidefinite programming plays a central role in the design of efficient algorithms and 
in dealing with $\NP$-hardness. For many fundamental problems, the best known (and sometimes provably best possible) approximation
algorithms are achieved via  relaxations based on semidefinite programs~\cite{GW95,Karger98:jacm,Raghavendra08:everycsp,Arora09:jacm,KT17,KawarabayashiTY24}.
In this work, we focus on computational problems of the following general form:  Given two structures (say, two digraphs) $\X$ and $\A$, is there a homomorphism from $\X$ to $\A$?
A plethora of different computational problems---in particular, those involving satisfiability of constraints---can be cast in this form.
The semidefinite programming paradigm is naturally applicable to this type of problems, and it yields relaxations that are \emph{robust to noise}: They are able to find a near-satisfying assignment even when the instance is almost---but not perfectly---satisfiable~\cite{Barto16:sicomp} (see also~\cite{bgs_robust23stoc}). On the other hand, 
certain homomorphism problems can be solved exactly in polynomial time but are inherently \emph{fragile to noise}---the primary example being systems of linear equations, which are tractable via Gaussian elimination but 
whose noisy version is $\NP$-hard~\cite{Hastad01}. Problems that behave like linear equations are hopelessly stubborn against the semidefinite programming model~\cite{Schoenebeck08,Tulsiani09:stoc,Chan15:jacm}.
It is then natural, in the context of homomorphism problems, to consider stronger versions of semidefinite programming relaxations that are equipped with a built-in linear-equation solver.

Consider a homomorphism\footnote{Letting $\Vset(\X)$ and $\Eset(\X)$ denote the vertex and edge sets of $\X$, a map $f:\Vset(\X)\to\Vset(\A)$ is a homomorphism if $(f(x),f(y))\in\Eset(\A)$ whenever $(x,y)\in\Eset(\X)$. The expression ``$\X\to\A$'' denotes the existence of a homomorphism.} $f:\X\to\A$. Letting $\lvert \Vset(\X)\rvert =p$ and $\lvert \Vset(\A)\rvert =n$, we can encode $f$ in a $pn\times pn$ matrix $M_f$ containing blocks of size $n\times n$, where the blocks are indexed by pairs of vertices of $\X$, and the entries in a block by pairs of vertices of $\A$.
For $x,y\in\Vset(\X)$ and $a,b\in\Vset(\A)$,
the $(a,b)$-th entry of the $(x,y)$-th block
is $1$ if $a=f(x)$ and $b=f(y)$, and $0$ otherwise. Let us explore the structure of $M_f$. Each block has nonnegative entries summing up to $1$,
and diagonal blocks are diagonal matrices. Since $f$ is a homomorphism, 
the $(a,b)$-th entry of the $(x,y)$-th block
is $0$ when $(x,y)\in\Eset(\X)$ and $(a,b)\not\in\Eset(\A)$. Finally, $M_f$ is positive semidefinite since it is symmetric and, for a $pn$-vector $\bv$, it satisfies $\bv^TM_f\bv
=(\sum_{x}v_{x,f(x)})^2\geq 0$.
The \emph{standard semidefinite programming relaxation} (SDP) of the homomorphism problem ``\emph{$\X\to\A$?}'' consists in 
looking for a real matrix $M$ with the properties described above.
We write $\SDP(\X,\A)=\YES$ if such a matrix $M$ exists.

Any \emph{Constraint Satisfaction Problem} (CSP)
may be expressed as the homomorphism problem of checking whether an \emph{instance} structure $\X$ homomorphically maps to a \emph{template} structure $\A$. Up to polynomial-time equivalence, $\X$ and $\A$ can be assumed to be digraphs without loss of generality~\cite{Feder98:monotone}. (As was shown in~\cite{BG21}, a similar fact holds for the promise version of CSP, which we shall encounter in a while.)
The power of semidefinite programming in the realm of CSPs is well understood: The CSPs solved by SDP are exactly those having \emph{bounded width}~\cite{Feder98:monotone,Barto16:sicomp}. Crucially, for CSPs, boosting SDP via the so-called \emph{lift-and-project technique}~\cite{Laurent03} does not increase its power:  Any semidefinite programming relaxation of polynomial size---in particular, any constant number of rounds of the \emph{Lasserre ``Sum-of-Squares'' hierarchy}~\cite{Lasserre02}---solves precisely the same CSPs as SDP~\cite{Barto16:sicomp,tz18}. 
The positive resolution of Feder--Vardi CSP Dichotomy Conjecture~\cite{Feder98:monotone} by Bulatov~\cite{Bulatov17:focs} and Zhuk~\cite{Zhuk20:jacm} 
implies that any tractable CSP is a (nontrivial) combination of $(i)$ bounded-width CSPs and $(ii)$ CSPs that can simulate linear equations (which have unbounded width).
The aim to find a \emph{universal solver} for all tractable CSPs has then driven a new generation of algorithms that combine $(i)$ techniques suitable for exploiting bounded width with $(ii)$ variants of Gaussian elimination (which solves linear equations). This line of work was pioneered by~\cite{BG19,bgwz20}, with the description of the algorithm BA
mixing a linear-programming-based relaxation with Gaussian elimination. 
Variants of this algorithm
were later considered in~\cite{cz23sicomp:clap,OConghaileC22:mfcs,Dalmau24:lics}.

The algorithm we propose in this work (which we call SDA) can be described as follows.
First, notice that the matrix $M_f$ encoding a homomorphism $f:\X\to\A$ has entries in $\{0,1\}$, and all of the properties of $M_f$ highlighted above are in fact linear equations, with the exception of the nonnegativity of its entries and the positive semidefiniteness. Hence, a different relaxation can be obtained by looking for a matrix $M'$ that respects the linear conditions, and whose entries are integers. We end up with a linear Diophantine system, that can be solved efficiently through integer variants of Gaussian elimination, see~\cite{schrijver1998theory}. 
We write $\SDA(\X,\A)=\YES$ if both $M$ and $M'$ exist, and a technical refinement condition constraining the
supports of $M$ and $M'$ holds (see Section~\ref{sec_preliminaries} for the formal definition of the algorithm).\footnote{The ``A'' in SDA stands for ``affine'' integer programming, the name by which the CSP relaxation based on
linear Diophantine equations is sometimes referred to in the literature~\cite{BBKO21}.} 

The first main goal of our work is to \emph{introduce a technique based on the spectral theory of association schemes 
for the analysis of this relaxation model}.
Our approach aims to describe how the algorithm exploits the symmetry of the problem under relaxation. To that end, we gradually refine and abstract the way symmetry is expressed. 
Starting from \emph{automorphisms}, which capture symmetry of $\X$ and $\A$,
we lift the analysis to 
the \emph{orbitals} of $\X$ and $\A$ under the action of the automorphism groups and, finally, we endow the orbitals with the algebraic structure of \emph{association schemes}.
The progressively more abstract language for expressing the symmetry of the problem yields a progressively cleaner description of the impact of symmetry on the relaxation.
For the SDP part of SDA,
the abstraction process 
``{automorphisms} $\to$ {orbitals} $\to$ {association schemes}''
may be viewed
in purely linear-algebraic terms, as the \emph{quest for a convenient (i.e., low-dimensional) vector space where the output of the algorithm lives, and a suitable basis for this space}.
The last stage of this metamorphosis of symmetry discloses a new algebraic perspective on the 
relaxation. 
In particular, for certain classes of digraphs, association schemes allow
turning SDP into a linear program. On a high level, this is an instance of a general \emph{invariant-theoretic} phenomenon:
The presence of a rich group of symmetries makes it possible to reduce the size of semidefinite programs~\cite{KlerkPS07,gatermann2004symmetry,kanno2001group} and, in certain cases, to describe their feasible regions in terms of linear inequalities~\cite{delsarte1973algebraic,Goemans99:computing}, see also~\cite{sturmfels2008algorithms,derksen2015computational}.
The non-convex nature of Diophantine equations makes the linear part of SDA process the symmetry of the inputs in a quite different way. We exploit the dihedral structure of the automorphism group of cycles to show that each associate in their scheme can be assigned an integral matrix with a small support; this, in turn, can be used to produce a solution $M'$ to the linear system.

This approach allows for a direct transfer of the results available in algebraic combinatorics on association schemes to the study of relaxations of homomorphism problems. 
For example, the explicit expression for the \emph{character table}
of a specific scheme known as the \emph{Johnson scheme} shall be crucial for establishing a lower bound against the SDA model.
One peculiarity of this framework is that \emph{it is not forgetful of the structure of the instance $\X$}. This  contrasts with the techniques for  describing relaxations of CSPs~\cite{ktz15:sicomp,Barto16:sicomp,tz17:sicomp,Dalmau18:jcss} based on the \emph{polymorphic approach}~\cite{Jeavons97:closure,Jeavons98:algebraic,Bulatov05:classifying}, whose gist is that the complexity of a CSP depends on the identities satisfied by the \emph{polymorphisms} of the CSP template $\A$~\cite{BOP18}. 
The polymorphic approach yields elegant characterisations of the power of some relaxations, in the sense that a CSP is solved by a certain algorithm if and only if its polymorphisms satisfy identities typical of the algorithm. (As established in~\cite{BBKO21}, a similar approach also works for the promise version of CSP that we shall discuss shortly.) These ``instance-free'' characterisations rely on having access to both the identities typical of the algorithm---not available in the case of SDP and, thus, SDA---and a succinct description of the polymorphisms of the template---which is missing in the case of the approximate homomorphism problems we shall see next.
In contrast, the description based on association schemes does take the structure of the instance into account, which results in a higher control over the behaviour of the algorithm on certain highly symmetric instances.

The second main goal of our work is to apply the framework of association schemes to \emph{obtain an unconditional lower bound against SDA} (and, a fortiori, against SDP). We consider the \emph{Approximate Graph Homomorphism} problem (AGH): Given two (undirected) graphs $\A$ and $\B$ such that $\A\to\B$ and an instance $\X$, distinguish between the cases $(i)$ $\X\to\A$ and $(ii)$ $\X\not\to\B$.\footnote{This is the \emph{decision} version of the problem. In the \emph{search} version, the goal is to find an explicit homomorphism $\X\to\B$ assuming that $\X\to\A$. The former version reduces to the latter, so our non-solvability result applies to both.} 
This problem is commonly studied in the context of \emph{Promise CSPs}~\cite{AGH17,BG21,BBKO21,KOWZ22}, and we shall thus denote it by $\PCSP(\A,\B)$.
Observe that $\PCSP(\A,\B)$ is well defined for any pair of digraphs (and, in fact, relational structures) $\A\to\B$.
If we let $\A=\K_n$ (the $n$-clique) and $\B=\K_{n'}$ where $n\leq n'$, AGH specialises to the \emph{Approximate Graph Colouring} problem (AGC): Distinguish whether a given graph is $n$-colourable or not even $n'$-colourable. The computational complexity of these problems is a long-standing open question. In contrast, the complexity of the non-approximate versions of AGC and AGH
(i.e., the cases $n=n'$ and $\A=\B$, respectively) was already classified by Karp~\cite{Karp72} and Hell--Ne\v{s}et\v{r}il~\cite{HellN90}, respectively.
In 1976, Garey and Johnson conjectured that AGC is always $\NP$-hard if $3\leq n$ (the case $n=2$ reduces to $2$-colouring and is thus tractable).
\begin{conj*}[\cite{GJ76}]
    Let $3\leq n\leq n'$ be integers. Then $\PCSP(\K_n,\K_{n'})$ is $\NP$-hard.
\end{conj*}
\noindent More recently, Brakensiek and Guruswami
proposed the stronger conjecture that even AGH may always be $\NP$-hard except in trivial cases (if either $\A$ or $\B$ has a loop or is bipartite, the problem is trivial or reduces to $2$-colouring).
\begin{conj*}[\cite{BG21}]
    Let $\A,\B$ be non-bipartite loopless undirected graphs such that $\A\to\B$. Then $\PCSP(\A,\B)$ is $\NP$-hard.
\end{conj*}
Among the several papers making progress on the two conjectures above, we mention~\cite{KOWZ22,KhannaLS00,Guruswami04:sidma,BrakensiekG16,BBKO21,Huang13,Khot01,Dinur09:sicomp,Braverman21:focs,BGS23}. However, they both remain wide open in their full generality.
Given the apparent ``hardness of proving hardness'' surrounding these problems, significant efforts have been directed towards showing inapplicability of specific algorithmic models, following an established line of work on lower bounds against relaxations, e.g.,~\cite{Arora06:toc,Chan15:jacm,Ghosh18:toc,Kothari22:sicomp,Lee15:stoc,Chan16:jacm-lp,Tulsiani09:stoc,Berkholz17:soda}.
Non-solvability of AGC via sublinear levels of local consistency and via linear Diophantine equations was proved in~\cite{Atserias22:soda} and~\cite{cz23soda:aip}, respectively.
It was shown in~\cite{Karger98:jacm} that the technique of vector colouring, based on a semidefinite program akin to Lov\'asz's orthonormal representation~\cite{lovasz1979shannon}, is inapplicable to solving AGC.
It follows from~\cite{GS20:icalp,Khot18:focs-pseudorandom} that polynomial levels of the Sum-of-Squares hierarchy (and, in particular, SDP) are also not powerful enough to solve AGC.
Very recently,~\cite{cz23stoc:ba} improved on the result in~\cite{cz23soda:aip} by proving non-solvability of AGC via constant levels of the 
BA hierarchy, obtained by applying the \emph{lift-and-project} technique to the BA relaxation of~\cite{bgwz20}.
By leveraging the framework of association schemes, we establish that AGH is not solved by SDA.
\begin{thm}
\label{thm_main_SDA_no_solves_approximate_homomorphism}
Let $\A,\B$ be non-bipartite loopless undirected graphs such that $\A\to\B$. Then
$\SDA$ does not solve $\PCSP(\A,\B)$.
\end{thm}
The improvement on the state of the art is twofold:  Theorem~\ref{thm_main_SDA_no_solves_approximate_homomorphism} yields $(i)$ the first non-solvability result for the whole class of problems AGH, as opposed to the subclass AGC, and $(ii)$  
the first lower bound against the combined ``SDP + linear equations'' model (which is strictly stronger than both models individually).
Via Raghavendra's framework~\cite{Raghavendra08:everycsp}, the (SDP part of the) integrality gap in Theorem~\ref{thm_main_SDA_no_solves_approximate_homomorphism} directly yields a conditional hardness-of-approximation result for AGH: Assuming Khot's Unique Games Conjecture~\cite{Khot02stoc} and $\operatorname{P}\neq\operatorname{NP}$, AGH is not solved by any polynomial-time \emph{robust} algorithm.

\paragraph{Related work on association schemes}
The Johnson scheme and other association schemes such as the \emph{Hamming scheme}
have appeared in the analysis of the performance ratio of Goemans--Williamson Max-Cut algorithm~\cite{GW95} based on semidefinite programming, see~\cite{Goemans99:computing,Karloff99,alon2000bipartite}. 
In~\cite{MekaPW15}, certain spectral properties of the Johnson scheme were used to obtain lower bounds against the Positivestellensatz
proof system (and, thus, against the Sum-of-Squares hierarchy) applied to the planted clique problem, see also~\cite{DeshpandeM15}.

\paragraph{Structure of the paper}
In Section~\ref{sec_preliminaries}, we formally define the algorithms used in this work and list some useful preliminary observations about them. The analysis of how the algorithms process the symmetry of the input structures begins in Section~\ref{sec_automorphisms_and_orbitals}, which provides a basis for the space of symmetry-invariant matrices output by the algorithms in terms of the orbitals of the input structures. Leveraging the theory of association schemes, Section~\ref{sec_association_schemes} describes an alternative basis, which gives easier access to the spectral properties of the matrices involved in the relaxations. Sections~\ref{sec_johnson_scheme} and~\ref{sec_cycle_scheme} investigate two specific association schemes---the Johnson scheme and the cycle scheme, respectively---whose properties are then used in Section~\ref{subsec_lower_bound_on_SDA} to conclude the proof of Theorem~\ref{thm_main_SDA_no_solves_approximate_homomorphism}. Finally, in Section~\ref{section_incomparability}, we compare the SDA algorithm with the BA hierarchy from~\cite{BG19,cz23stoc:ba}, and we show that our lower bound is incomparable with the one in~\cite{cz23stoc:ba}, as SDA---and, in fact, even a weaker version of SDP described in Subsection~\ref{subsec_SDPeps}---is not less powerful than the BA hierarchy in a technical sense (Subsection~\ref{subsec_SDPeps_vs_BAk}).

\paragraph{Notation}
We let $\N$ be the set of positive integers, while $\N_0=\N\cup\{0\}$. For $t\in\N$, we let $[t]=\{1,\dots,t\}$.
We view vectors in $\R^t$ as column vectors, but sometimes write them as tuples for typographical convenience.
We denote by $I_t$ and $J_t$ the $t\times t$ identity and all-one matrices, by $O_{t,t'}$ the $t\times t'$ all-zero matrix, and by $\bone_t$ and $\bzero_t$ the all-one and all-zero vectors of length $t$. Indices shall sometimes be omitted when clear from the context. We denote by $\be_i$ the $i$-th standard unit vector of length $t$ (which shall be clear from the context); i.e., the vector in $\R^t$ all of whose entries are $0$ except the $i$-th entry that is $1$. 
Given a field $\mathbb{F}$ and a set $\mathscr{V}$ of vectors in $\mathbb{F}^t$, $\spann_{\mathbb{F}}(\mathscr{V})$ is the set of linear combinations over $\mathbb{F}$ of the vectors in $\mathscr{V}$. We write $\spann(\mathscr{V})$ for $\spann_{\mathbb{R}}(\mathscr{V})$.

A matrix is \emph{Boolean} if its entries are in $\{0,1\}$.
Given a real matrix $M$, we write $M\geq 0$ if $M$ is \emph{entrywise nonnegative}, and we write $M\succcurlyeq 0$ if $M$ is \emph{positive semidefinite} (i.e., if $M$ is symmetric and has a nonnegative spectrum).
For two matrices $M=(m_{ij})$ and $M'$ of size $m\times n$ and $m'\times n'$, respectively, we let their \emph{Kronecker product} be the $mm'\times nn'$ block-matrix $M\otimes M'$ whose $(i,j)$-th block, for $i\in[m]$ and $j\in [n]$, is the matrix $m_{ij}M'$. If $M$ and $M'$ have equal size, we let 
$M\circ M'$ denote their \emph{Schur product} (i.e., their entrywise product, also known as \emph{Hadamard product}).
We shall often use the fact that $(M\otimes M')(N\otimes N')=MN\otimes M'N'$, provided that the products are well defined (see~\cite{horn1994topics}). 
The \emph{support} of $M$, denoted by
$\supp(M)$, is the set of indices of nonzero entries of $M$; for two matrices $M,M'$ of equal size, we write $M\triangleleft M'$ to denote that $\supp(M)\subseteq\supp(M')$.
Given a digraph $\X$, we let $\adj(\X)$ and $\Aut(\X)$ denote the \emph{adjacency matrix} and the \emph{automorphism group} of $\X$, respectively. We view undirected graphs as digraphs, by turning each undirected edge $\{x,y\}$ into a pair of directed edges $(x,y)$ and $(y,x)$.

\section{The algorithms}
\label{sec_preliminaries}
In the literature on CSPs,
it is customary to define SDP in terms of systems of vectors satisfying certain orthogonality requirements~\cite{Barto16:sicomp,tz18,Raghavendra08:everycsp,bgs_robust23stoc,cz23soda:minions}. 
We now present this standard \emph{vector formulation} of $\SDP$ and define the augmented algorithm $\SDA$ along the same lines. Then, we describe an alternative formulation of both relaxations, which is more suitable for our purposes, and establish the equivalence between the two.
To enhance the paper's readability and help the reader reach the technical core more quickly, we defer 
the proofs of the results in this 
section to Appendix~\ref{sec:app}.

Let $\X$ and $\A$ be two digraphs, and label their vertex sets as $\Vset(\X)=[p]$ and $\Vset(\A)=[n]$ for some $p,n\in\N$.
We introduce a vector variable $\blambda_{x,a}$ taking values in $\R^{pn}$ for all vertices $x\in \Vset(\X)$, $a\in \Vset(\A)$, and we set $\blambda_{x,\A}=\sum_{a\in \Vset(\A)}\blambda_{x,a}$.
Consider the system
\labeltext{$\SDP_1$}{SDP1}
\labeltext{$\SDP_2$}{SDP2}
\labeltext{$\SDP_3$}{SDP3}
\labeltext{$\SDP_4$}{SDP4}
\begin{align}
\label{eqns_SDP}
\tag{SDP}
\boxed{
\begin{array}{llllll}
\mbox{($\SDP_1$)}\qquad &\blambda_{x,a}\cdot\blambda_{y,b}&\geq& 0 &\hspace*{.1cm} & \forall x,y\in \Vset(\X),\;  a,b\in \Vset(\A)\\
\mbox{($\SDP_2$)}\qquad &\blambda_{x,a}\cdot\blambda_{x,b}&=& 0 &\hspace*{.1cm} & \forall x\in \Vset(\X),\;  a\neq b\in \Vset(\A)\\
\mbox{($\SDP_3$)}\qquad &\blambda_{x,a}\cdot\blambda_{y,b}&=& 0 & \hspace*{.1cm} &\forall (x,y)\in \Eset(\X),\;  (a,b)\in \Vset(\A)^2\setminus\Eset(\A)\\
\mbox{($\SDP_4$)}\qquad &\blambda_{x,\A}\cdot\blambda_{y,\A}&=&1 & \hspace*{.1cm} &\forall x,y\in \Vset(\X).
\end{array}
}
\end{align}

\noindent Note that ($\SDP_4$) forces all vectors $\blambda_{x,\A}$ to be equal.
We say that SDP applied to $\X,\A$ \emph{accepts}, and we write $\SDP(\X,\A)=\YES$, if the system~\eqref{eqns_SDP} has a solution.

In order to augment $\SDP$ with the linear Diophantine part,
we introduce additional
variables $\mu_{x,a}$ taking values in $\Z$ for all vertices $x\in \Vset(\X),a\in \Vset(\A)$, and variables
$\mu_{\bx,\ba}$ taking values in $\Z$ for all directed edges $\bx\in E(\X),\ba\in E(\A)$, and we consider the equations
\labeltext{$\AIP_1$}{AIP1}
\labeltext{$\AIP_2$}{AIP2}
\begin{align}
\label{eqns_AIP}
\tag{AIP}
\boxed{
\begin{array}{llllll}
\mbox{($\AIP_1$)}\qquad &\displaystyle\sum_{a\in \Vset(\A)} \mu_{x,a}\ &=& 1 &  &\forall x\in \Vset(\X)\\
\mbox{($\AIP_2$)}\qquad &\displaystyle\sum_{\ba\in E(\A),\; a_i=a} \mu_{\bx,\ba}\ &=& \mu_{x_i,a} & &\forall \bx\in E(\X),\; i\in \{1,2\},\; a\in \Vset(\A).
\end{array}
}
\end{align}
The $\SDA$ relaxation 
consists in $(i)$ searching a solution to~\eqref{eqns_SDP}, $(ii)$ discarding assignments having zero probability, and then $(iii)$ searching a solution to~\eqref{eqns_AIP}.\footnote{The name ``AIP'' stands for ``affine integer programming'', see~\cite{BBKO21}.} Formally, 
we say that $\SDA$ applied to $\X$ and $\A$ \emph{accepts}, and we write $\SDA(\X,\A)=\YES$, if the system~\eqref{eqns_SDP} admits a solution $\blambda$ and the system~\eqref{eqns_AIP} admits a solution $\mu$ such that the following \emph{refinement condition} holds:
\begin{align}
\label{eqn_refinement_condition_SDA}
\tag{ref}
\boxed{
\begin{array}{lllllll}
\|\blambda_{x,a}\|=0\;\;&\Rightarrow&\;\;\mu_{x,a}=0&\;\;\forall x\in \Vset(\X),\,a\in \Vset(\A)\\
\blambda_{x_1,a_1}\cdot\blambda_{x_2,a_2}=0\;\;&\Rightarrow&\;\;\mu_{\bx,\ba}=0&\;\;\forall \bx\in \Eset(\X),\, \ba\in \Eset(\A)
\end{array}
}
\end{align}
(where $\|\cdot\|$ is the Euclidean vector norm).

Given two digraphs $\A,\B$ such that $\A\to\B$, we say that SDP (resp., SDA) \emph{solves} $\PCSP(\A,\B)$ if, for any digraph $\X$, $\SDP(\X,\A)=\YES$ (resp., $\SDA(\X,\A)=\YES$) implies $\X\to\B$. It follows from the definitions of the algorithms that $\X\to\A$ always implies $\SDP(\X,\A)=\SDA(\X,\A)=\YES$.

It shall be convenient for our purposes to use a slightly modified, \textit{matrix formulation} of the relaxations, which allows for a more intuitive linear-algebraic view of the algorithms' behaviour. The equivalence of the two formulations is given in Proposition~\ref{prop_matrix_formulation_SDP_SDA} below.
\begin{defn}
A real $pn\times pn$ matrix $M$ is a \emph{relaxation matrix} for $\X,\A$ if $M$ satisfies the following requirements:
\begin{enumerate}
    \item[($r_1$)\labeltext{$r_1$}{r1}]
    $(\be_x\otimes I_n)^T M (\be_x\otimes I_n)$ is a diagonal matrix for each $x\in \Vset(\X)$;
    \item[($r_2$)\labeltext{$r_2$}{r2}]
    $(\be_x\otimes\be_a)^T M (\be_y\otimes\be_b)=0$
    whenever $(x,y)\in\Eset(\X)$ and $(a,b)\in \Vset(\A)^2\setminus\Eset(\A)$;
    \item[($r_3$)\labeltext{$r_3$}{r3}] $M(\be_x\otimes\bone_n)=M(\be_{y}\otimes\bone_n)$ for each $(x,y)\in \Vset(\X)^2$;
    \item[($r_4$)\labeltext{$r_4$}{r4}] $M^T(\be_x\otimes\bone_n)=M^T(\be_{y}\otimes\bone_n)$ for each $(x,y)\in \Vset(\X)^2$;
    \item[($r_5$)\labeltext{$r_5$}{r5}] $\bone_{pn}^TM\bone_{pn}=p^2$.
\end{enumerate}
Given a relaxation matrix $M$, we say that $M$ is an \emph{$\SDP$-matrix} for $\X,\A$ if $M\succcurlyeq 0$ and $M\geq 0$,
and we say that $M$ is an \emph{$\AIP$-matrix} for $\X,\A$ if all of its entries are integral.
\end{defn}
Notice that the definition of an $\SDP$-matrix captures the description of the matrix $M_f$ considered at the beginning of the Introduction to illustrate SDP. 
Indeed, viewing $M$ as a block matrix whose $n\times n$ blocks are indexed by pairs in $\Vset(\X)^2$, (\ref{r1}) states that diagonal blocks are diagonal matrices, (\ref{r2}) states that the supports of blocks corresponding to edges of $\X$ are included in $\Eset(\A)$, (\ref{r3}) and (\ref{r4}) state that the row-sum (resp. column-sum) vectors of blocks aligned horizontally (resp. vertically) are equal, and (\ref{r5}) is a normalisation condition.
Interestingly, the very similar definition of an $\AIP$-matrix is able to capture the linear Diophantine part of SDA.

For a square matrix $A$, the set of vectors $\bv$ for which the Rayleigh quotient $\bv^TA\bv$ is zero clearly includes the null space of $A$. The inclusion is in general strict, even assuming $A$ to be symmetric. However, the two sets coincide when $A$ is positive semidefinite (see~\cite[Obs.~7.1.6]{Horn2012matrix}). The next result
is essentially a specialisation of this fact to the conditions (\ref{r1})--(\ref{r5}) defining relaxation matrices.
Given a real $pn\times pn$ matrix $M$, consider the condition
\begin{enumerate}
    \item[($r_6$)\labeltext{$r_6$}{r6}]
    $(\be_x\otimes\bone_n)^TM(\be_y\otimes\bone_n)= 1$ for each $(x,y)\in \Vset(\X)^2$.
\end{enumerate}
\begin{prop}
\label{prop_equivalence_some_conditions_SDP}
Let $M$ be a real $pn\times pn$ matrix. Then
\begin{itemize}
    \item[$(i)$] the conditions \emph{(\ref{r3})}, \emph{(\ref{r4})}, and \emph{(\ref{r5})} imply the condition \emph{(\ref{r6})};
    \item[$(ii)$] if $M\succcurlyeq 0$, the conditions \emph{(\ref{r3})}, \emph{(\ref{r4})}, and \emph{(\ref{r5})} are equivalent to the condition \emph{(\ref{r6})}.
\end{itemize}
\end{prop}
By using Proposition~\ref{prop_equivalence_some_conditions_SDP}, we can
convert the vector formulation of the algorithms SDP and SDA into an equivalent formulation in terms of relaxation matrices. 
\begin{prop}
\label{prop_matrix_formulation_SDP_SDA}
    Let $\X,\A$ be digraphs. Then
    \begin{itemize}
        \item[$(i)$] $\SDP(\X,\A)=\YES$ if and only if there exists an $\SDP$-matrix for $\X,\A$;
        \item[$(ii)$] if $\X$ is loopless, $\SDA(\X,\A)=\YES$ if and only if there exist an $\SDP$-matrix $M$ and an $\AIP$-matrix $N$ for $\X,\A$ such that $N\circ ((I_p+\adj(\X))\otimes J_n)\;\;\triangleleft\;\; M$.
    \end{itemize}
\end{prop}
In particular, if $M$ is an $\SDP$-matrix for $\X,\A$, the corresponding vector formulation involves the vectors $\blambda_{x,a}$ consisting of the columns of a $pn\times pn$ matrix $L$ such that $M=L^TL$ is a \emph{Cholesky decomposition} of $M$. Since $M\succcurlyeq 0$, such a decomposition always exists.
Moreover, the requirement $N\circ ((I_p+\adj(\X))\otimes J_n)\triangleleft M$ in part $(ii)$ of Proposition~\ref{prop_matrix_formulation_SDP_SDA} captures the refinement condition~\eqref{eqn_refinement_condition_SDA} of the SDA algorithm.

In the remaining part of this section, we give a useful result on the behaviour of the relaxations defined above with respect to digraph homomorphisms, which can be easily derived by looking at the corresponding relaxation matrices.
Given two finite sets $R$ and $S$ and a function $f:R\to S$, we let $Q_f$ be the $\lvert R\rvert\times \lvert S\rvert$ matrix whose $(r,s)$-th entry is $1$ if $f(r)=s$, $0$ otherwise. 
The next lemma, whose trivial proof is omitted, lists some useful properties of $Q_f$.
\begin{lem}
\label{lem_basic_Q_f}
    Let $R$, $S$, $T$ be finite sets, and let $f:R\to S$, $g:S\to T$ be functions. Then
    \begin{itemize}
        \item $Q_f^T\be_r=\be_{f(r)}$ for each $r\in R$;
        \item $Q_f\be_s=\sum_{r\in f^{-1}(s)}\be_r$ for each $s\in S$;
        \item $Q_f\bone_{\lvert S\rvert}=\bone_{\lvert R\rvert}$;
        \item 
        $Q_fQ_g=Q_{g\circ f}$;
        \item 
        if $f$ is bijective, $Q_f$ is invertible and $Q_f^{-1}=Q_f^T=Q_{f^{-1}}$.
    \end{itemize}
\end{lem}
We now look at what happens when a relaxation is applied to two different pairs of inputs $(\X,\A)$ and $(\X',\A')$ such that $\X'\to\X$ and $\A\to\A'$. Expressing the two outputs in the form of relaxation matrices,
the next proposition shows that one output can be obtained from the other through Kronecker products of the matrices $Q_f$.

\begin{prop}
\label{prop_relaxation_matrix_preservation_homo}
    Let $\X,\X',\A,\A'$ be digraphs, let $f:\X'\to\X$ and $g:\A\to\A'$ be homomorphisms, and let $M$ be a relaxation matrix for $\X,\A$. Then
\begin{align}
\label{eqn_0106_1611}
    M^{(f,g)}\spaceeq(Q_f\otimes Q_g^T)M(Q_f^T\otimes Q_g)
\end{align}
is a relaxation matrix for $\X',\A'$. Furthermore, if $M$ is an $\SDP$-matrix (resp. $\AIP$-matrix) for $\X,\A$, then $M^{(f,g)}$ is an $\SDP$-matrix (resp. $\AIP$-matrix) for $\X',\A'$.
\end{prop}
The last part of Lemma~\ref{lem_basic_Q_f} states that the matrix $Q_f$ corresponding to a bijective function $f$ is orthogonal, in that $Q_f^{-1}=Q_f^T$. Note that the Kronecker product of orthogonal matrices is an orthogonal matrix. Therefore, if the homomorphisms $f$ and $g$ in Proposition~\ref{prop_relaxation_matrix_preservation_homo} are both bijective (for example, if they are isomorphisms), the linear operator $(\cdot)^{(f,g)}:M\mapsto M^{(f,g)}$ is an \emph{orthogonal transformation} with respect to the Frobenius inner product $\Frob{M}{N}=\Tr(M^TN)$. Indeed, letting $P=Q_f\otimes Q_g^T$, we have
\begin{align*}
    \Frob{M^{(f,g)}}{N^{(f,g)}}
    =
    \Tr(PM^TP^TPNP^T)
    =
    \Tr(P^TPM^TP^TPN)
    =
    \Tr(M^TN)
    =\Frob{M}{N}.
\end{align*}

A straightforward consequence of Proposition~\ref{prop_relaxation_matrix_preservation_homo} is that the algorithms are monotone with respect to the homomorphism preorder.

\begin{prop}
\label{prop_monotonicity_SDP_SDA_wrt_homomorphisms}
    Let $\X,\X',\A,\A'$ be digraphs such that $\X'\to\X$ and $\A\to\A'$. Then
    \begin{itemize}
        \item[$(i)$] $\SDP(\X,\A)=\YES$ implies $\SDP(\X',\A')=\YES$;
        \item[$(ii)$] if $\X$ is loopless, $\SDA(\X,\A)=\YES$ implies $\SDA(\X',\A')=\YES$.
    \end{itemize}
\end{prop}

\section{Automorphisms and orbitals}
\label{sec_automorphisms_and_orbitals}
The leitmotif of this work is the use of linear algebra to manipulate relaxation algorithms. 
In this section, we begin to explore how the symmetries of the input digraphs---expressed via their automorphism groups---affect the outputs---expressed via relaxation matrices.

As usual, we let $\X$ and $\A$ be two digraphs whose vertex sets have size $p$ and $n$, respectively.
If $\xi$ and $\alpha$ are automorphisms of $\X$ and $\A$, respectively, we may permute the rows and columns of a relaxation matrix $M$ according to $\xi$ and $\alpha$ and the result would still be a relaxation matrix.
By averaging over all pairs of automorphisms $(\xi,\alpha)$, we end up with a relaxation matrix that is invariant under automorphisms of $\X$ and $\A$.
Since the set of positive semidefinite, entrywise-nonnegative matrices is closed under simultaneous permutations of rows and columns and under convex combinations, the same can be done for $\SDP$-matrices.

We now formalise this observation.
The next definition captures the invariance property mentioned above. 
Recall the description of the matrices $Q_f$ and $M^{(f,g)}$ given in Section~\ref{sec_preliminaries}.
\begin{defn}
\label{defn_balancedness}
A real $pn\times pn$ matrix $M$ is \emph{balanced} for $\X,\A$ if 
\begin{align*}
    (Q_\xi\otimes Q_\alpha)M(Q_\xi^T\otimes Q_\alpha^T)
    \spaceeq
    M &&\mbox{ for each }\xi\in\Aut(\X),\;\alpha\in\Aut(\A).
\end{align*}
\end{defn}

\begin{prop}
\label{relaxation_matrix_can_be_balanced}
    Let $\X,\A$ be digraphs, let
$s=\lvert\Aut(\X)\rvert$ and $t=\lvert\Aut(\A)\rvert$, and let $M$ be a relaxation matrix for $\X,\A$. Then the matrix
\begin{align}
\label{eqn_03052023}
    \overline M
    \spaceeq
    \frac{1}{st}\sum_{\substack{\xi\in\Aut(\X)\\\alpha\in\Aut(\A)}}M^{(\xi,\alpha)}
\end{align}
is a balanced relaxation matrix for $\X,\A$. Furthermore, if $M$ is an $\SDP$-matrix for $\X,\A$, then so is $\overline M$.
\end{prop}
\begin{proof}
For any $\xi\in\Aut(\X)$ and $\alpha\in\Aut(\A)$, since automorphisms are, in particular, homomorphisms, we deduce from Proposition~\ref{prop_relaxation_matrix_preservation_homo} that $M^{(\xi,\alpha)}$ is a relaxation matrix for $\X,\A$.
Since the conditions (\ref{r1})--(\ref{r5}) are clearly preserved by taking convex combinations, it follows that $\overline M$ is a relaxation matrix for $\X,\A$, too. We are left to show that $\overline M$ is balanced. 
For $\xi_0\in\Aut(\X)$ and $\alpha_0\in\Aut(\A)$, using Lemma~\ref{lem_basic_Q_f}, we find
{
\allowdisplaybreaks
\begin{align*}
    (Q_{\xi_0}\otimes Q_{\alpha_0})\overline M(Q_{\xi_0}^T\otimes Q_{\alpha_0}^T)
    &\spaceeq
    \frac{1}{st}\sum_{\substack{\xi\in\Aut(\X)\\\alpha\in\Aut(\A)}}(Q_{\xi_0}\otimes Q_{\alpha_0})(Q_\xi\otimes Q_\alpha^T)M(Q_\xi^T\otimes Q_\alpha)(Q_{\xi_0}^T\otimes Q_{\alpha_0}^T)\\
    &\spaceeq
    \frac{1}{st}\sum_{\substack{\xi\in\Aut(\X)\\\alpha\in\Aut(\A)}}
    (Q_{\xi\circ\xi_0}\otimes Q_{\alpha^{-1}\circ\alpha_0})M(Q_{\xi_0^{-1}\circ\xi^{-1}}\otimes Q_{\alpha_0^{-1}\circ\alpha})\\
    &\spaceeq
    \frac{1}{st}\sum_{\substack{\xi\in\Aut(\X)\\\alpha\in\Aut(\A)}}
    (Q_{\xi\circ\xi_0}\otimes Q_{\alpha_0^{-1}\circ\alpha}^T)M(Q_{\xi\circ\xi_0}^T\otimes Q_{\alpha_0^{-1}\circ\alpha})\\
    &\spaceeq
    \frac{1}{st}\sum_{\substack{\xi'\in\Aut(\X)\\\alpha'\in\Aut(\A)}}
    (Q_{\xi'}\otimes Q_{\alpha'}^T)M(Q_{\xi'}^T\otimes Q_{\alpha'})
    \spaceeq
    \overline M,
\end{align*}
}as required (where the penultimate equality holds since $\Aut(\X)$ and $\Aut(\A)$ are groups).
If $M$ is an $\SDP$-matrix for $\X,\A$, the same holds for $M^{(\xi,\alpha)}$ for any $\xi,\alpha$ (by virtue of Proposition~\ref{prop_relaxation_matrix_preservation_homo}) and for $\overline M$ (since positive semidefiniteness and entrywise nonnegativity are preserved under convex combinations).
\end{proof}
The next result, obtained as a consequence of Proposition~\ref{relaxation_matrix_can_be_balanced}, is a symmetric version of Proposition~\ref{prop_matrix_formulation_SDP_SDA}.

\begin{prop}
\label{prop_matrix_char_SDP_new_form}
    Let $\X,\A$ be digraphs. Then
    \begin{itemize}
        \item[$(i)$] $\SDP(\X,\A)=\YES$ if and only if there exists a balanced $\SDP$-matrix for $\X,\A$;
        \item[$(ii)$] if $\X$ is loopless, $\SDA(\X,\A)=\YES$ if and only if there exist a balanced $\SDP$-matrix $M$ and an $\AIP$-matrix $N$ for $\X,\A$ such that $N\circ ((I_p+\adj(\X))\otimes J_n)\;\;\triangleleft\;\; M$.
    \end{itemize}
\end{prop}
\begin{proof}
    The result immediately follows by combining Proposition~\ref{prop_matrix_formulation_SDP_SDA} with Proposition~\ref{relaxation_matrix_can_be_balanced} and observing that, if $M$ is an $\SDP$-matrix, the matrix $\overline M$ defined in~\eqref{eqn_03052023} satisfies $M\triangleleft\overline M$.
\end{proof}

As a result of Proposition~\ref{prop_matrix_char_SDP_new_form}, the output of $\SDP$ (and of the $\SDP$-part of $\SDA$) may be assumed to be balanced without loss of generality.\footnote{Clearly, the same is not true for $\AIP$, as integral matrices are not closed under convex combinations.}
In linear-algebraic terms, it follows that, instead of studying the outputs of $\SDP$ in $\R^{pn\times pn}$ with the basis of standard unit matrices $\be_i\be_j^T$ (as we have implicitly done so far), we may work without loss of generality in the real vector space $\mathscr{L}$ of balanced matrices for $\X,\A$. (The fact that $\mathscr{L}$ is a real vector space easily follows from Definition~\ref{defn_balancedness}.) As we see next, the concept of orbitals provides a natural basis for the space $\mathscr{L}$.

Take a digraph $\X$, and consider the action of the group $\Aut(\X)$ onto the set $\Vset(\X)^2$ given by $(x,y)^\xi=(\xi(x),\xi(y))$ for $\xi\in\Aut(\X)$, $x,y\in \Vset(\X)$. An \emph{orbital} of $\X$ is an orbit of $\Vset(\X)^2$ with respect to this action; i.e., it is a minimal subset of $\Vset(\X)^2$ that is invariant under the action. We let $\Orb(\X)$ be the set of orbitals of $\X$. 
Given an orbital $\omega\in\Orb(\X)$, we let $R_\omega$ be the $p\times p$ matrix whose $(x,y)$-th entry is $1$ if $(x,y)\in\omega$ and $0$ otherwise.
Orbitals provide an alternative description of balanced matrices: A block matrix $M$ is balanced for $\X,\A$ if and only if the block structure of $M$ is constant over the orbitals of $\X$, and each block is constant over the orbitals of $\A$. As stated next, it follows that we can find a basis for $\mathscr{L}$ 
by taking Kronecker products of the matrices $R_\omega$. We shall see later that a different basis for the same space may be found under certain conditions using the theory of association schemes.
\begin{prop}
\label{prop_orbital_matrix_of_balanced_matrix}
    Let $\X,\A$ be digraphs, and let $\mathscr{L}$ be the real vector space of balanced matrices for $\X,\A$. Then the set $\mathscr{R}=\{R_\omega\otimes R_{\tilde\omega}:\omega\in\Orb(\X),\tilde\omega\in\Orb(\A)\}$ forms a basis for $\mathscr{L}$.
\end{prop}
\begin{proof}
For $\omega\in\Orb(\X)$, $\xi\in\Aut(\X)$, and $x,y\in\Vset(\X)$, we have
\begin{align*}
    \be_x^TQ_\xi R_\omega Q_\xi^T\be_y
    &\spaceeq
    \be_{\xi(x)}^T R_\omega\be_{\xi(y)}
    \spaceeq
    \left\{
    \begin{array}{cccc}
         1&\mbox{ if }(x,y)^\xi\in\omega  \\
         0&\mbox{ otherwise} 
    \end{array}
    \right.
    \spaceeq
    \left\{
    \begin{array}{cccc}
         1&\mbox{ if }(x,y)\in\omega  \\
         0&\mbox{ otherwise} 
    \end{array}
    \right.\\
    &\spaceeq
    \be_x^TR_\omega\be_y,
\end{align*}
which means that $Q_\xi R_\omega Q_\xi^T=R_\omega$. Similarly, $Q_\alpha R_{\tilde\omega} Q_\alpha^T=R_{\tilde\omega}$ for each $\tilde\omega\in\Orb(\A)$ and $\alpha\in\Aut(\A)$. As a consequence, we find
\begin{align*}
    (Q_\xi\otimes Q_\alpha)(R_\omega\otimes R_{\tilde\omega})(Q_\xi^T\otimes Q_\alpha^T)
    \spaceeq
    (Q_\xi R_\omega Q_\xi^T)\otimes(Q_\alpha R_{\tilde\omega}Q_{\alpha}^T)
    \spaceeq
    R_\omega\otimes R_{\tilde\omega},
\end{align*}
thus showing that the matrix $R_\omega\otimes R_{\tilde\omega}$ is balanced for $\X,\A$. It follows that $\mathscr{R}\subseteq\mathscr{L}$.
Since the orbits of a group action partition the underlying set, we can write $\Vset(\X)^2$ as the disjoint union of the orbitals:
\begin{align*}
    \Vset(\X)^2=\bigsqcup_{\omega\in\Orb(\X)}\omega.
\end{align*}
Hence, $\sum_{\omega\in\Orb(\X)}R_\omega= J_p$. Similarly, $\sum_{\tilde\omega\in\Orb(\A)}R_{\tilde\omega}= J_n$, and 
\begin{align*}
    \sum_{\substack{\omega\in\Orb(\X)\\\tilde\omega\in\Orb(\A)}}R_\omega\otimes R_{\tilde\omega}
    \spaceeq
    J_p\otimes J_n
    \spaceeq
    J_{pn}.
\end{align*}
Therefore, $\mathscr{R}$ consists of Boolean matrices summing up to the all-one matrix, and it is thus a linearly independent set. Given $M\in\mathscr{L}$, $\omega\in\Orb(\X)$, and $\tilde\omega\in\Orb(\A)$, let $v_{\omega\tilde\omega}=(\be_x\otimes\be_a)^T M(\be_y\otimes\be_b)$ for some $(x,y)\in\omega$, $(a,b)\in\tilde\omega$. This definition is well posed by virtue of
Definition~\ref{defn_balancedness}, and it guarantees that
\begin{align}
\label{eqn_1334_0504}
    M=\sum_{\substack{\omega\in\Orb(\X)\\\tilde\omega\in\Orb(\A)}}v_{\omega\tilde\omega}\,R_\omega\otimes R_{\tilde\omega}.
\end{align}
It follows that $\spann(\mathscr{R})=\mathscr{L}$, which concludes the proof.
\end{proof}
It follows from Proposition~\ref{prop_orbital_matrix_of_balanced_matrix} that, given a balanced matrix $M$, there exists a unique list of coefficients $v_{\omega\tilde\omega}$ satisfying the equation~\eqref{eqn_1334_0504}.
We shall refer to the $\lvert \Orb(\X)\rvert \times\lvert \Orb(\A)\rvert$ matrix $V=(v_{\omega\tilde\omega})$ as the \emph{orbital matrix} of $M$. 
Expressing a balanced matrix $M$ in the new basis $\mathscr{R}$ rather than in the standard basis for $\R^{pn\times pn}$ is especially convenient when $\X$ and $\A$ are highly symmetric. Indeed, if $\Aut(\X)$ and $\Aut(\A)$ are large, $\Orb(\X)$ and $\Orb(\A)$ are small. Working with $\mathscr{R}$ allows then compressing the information of the $pn\times pn$ matrix $M$ in the smaller $\lvert\Orb(\X)\rvert\times\lvert\Orb(\A)\rvert$ orbital matrix $V$. However, if we want to make use of $V$ to certify acceptance of $\SDP$, we need to be able to check if $M$ is an $\SDP$-matrix by only looking at $V$. 
While lifting the requirements defining an $\SDP$-matrix to the orbital matrix, it should come with little surprise that the crucial one is positive semidefiniteness: How to translate the fact that $M\succcurlyeq 0$ into a condition on $V$?
We shall see in the next section that the key for recovering the spectral properties of $M$ from the orbital matrix is to endow the set of orbitals with a certain algebraic structure.

\section{Association schemes}
\label{sec_association_schemes}
Our strategy for gaining access to the spectral properties of a balanced matrix (in particular, its positive semidefiniteness) from the corresponding orbital matrix consists in studying the orbitals of a given digraph algebraically, via the concept of association schemes.
\begin{defn}
\label{defn_association_scheme}
An \emph{association scheme} is a set $\mathscr{S}=\{S_0,S_1,\dots,S_d\}$ of $p\times p$ Boolean matrices satisfying
\labeltext{$s_1$}{s1}
\labeltext{$s_2$}{s2}
\labeltext{$s_3$}{s3}
\labeltext{$s_4$}{s4}
\labeltext{$s_5$}{s5}
\begin{align*}
    \begin{array}{llllll}
        (s_1)\;\; S_0=I_p
      &\hspace{.3cm} (s_2)\;\; \sum_{i=0}^d S_i=J_p
      &\hspace{.3cm} (s_3)\;\; S_i^T\in\mathscr{S} \;\;\forall i\\[5pt]
        (s_4)\;\; S_iS_j\in\spann_{\mathbb{C}}(\mathscr{S})\;\; \forall i,j
      &\hspace{.3cm} (s_5)\;\; S_iS_j=S_jS_i\;\; \forall i,j.
    \end{array}
\end{align*}
\end{defn}
Association schemes were introduced by Bose and Nair~\cite{bose1939partially} and Bose and Shimamoto~\cite{bose1952classification} in the context of statistical design of experiments, but the root of the theory can be traced back to the work of Frobenius, Schur, and Burnside on representation theory of finite groups, see~\cite{bannai_ito_current_research}. Indeed, if all $S_i$ in Definition~\ref{defn_association_scheme} are permutation matrices, $\mathscr{S}$ is a finite group; association schemes allow developing a theory of symmetry that generalises character theory for group representations.
Later, Delsarte's work in algebraic coding theory~\cite{delsarte1973algebraic} initiated the study of association schemes as a separate area 
in the domain of algebraic combinatorics. 

The \emph{Bose--Mesner algebra} $\mathfrak{B}$ of $\mathscr{S}$ is the vector space $\spann_{\mathbb{C}}(\mathscr{S})$, which consists of all complex linear combinations of the matrices in $\mathscr{S}$ (see~\cite{Bose59}). Since the matrices in $\mathscr{S}$ are Boolean and satisfy (\ref{s2}), they form a basis for $\mathfrak{B}$. Notice also that the set $\mathscr{S}\cup \{O_{p,p}\}$ is closed under the Schur product, and so is $\mathfrak{B}$. Moreover, the matrices in $\mathscr{S}$ are Schur-orthogonal and Schur-idempotent, in that $S_i\circ S_j$ equals $S_i$ when $i=j$, and equals $O_{p,p}$ otherwise. Hence, we have the following.\footnote{Facts~\ref{fact_first_bases_Bose_Mesner} and~\ref{fact_second_bases_Bose_Mesner} can be found in any of the references~\cite{bannai_ito_1984,godsil2016erdos,brouwer1989distance,godsil2010associationSchemes}.}
\begin{fact}
\label{fact_first_bases_Bose_Mesner}
Let $\mathscr{S}$ be an association scheme.
Then $\mathscr{S}$ forms a Schur-orthogonal basis of Schur-idempotents for its Bose--Mesner algebra $\mathfrak{B}$.
\end{fact}
\noindent
Now, by (\ref{s4}), $\mathfrak{B}$ is also closed under the standard matrix product; in other words, it is a matrix algebra, thus justifying the name. It turns out that a different basis exists for $\mathfrak{B}$, whose members enjoy similar
properties to those for the basis $\mathscr{S}$, but with a different product being involved.
\begin{fact}
\label{fact_second_bases_Bose_Mesner}
Let $\mathscr{S}$ be an association scheme.
Then there exists an orthogonal basis $\mathscr{E}=\{E_0,E_1,\dots,E_d\}$ of Hermitian idempotents\footnote{I.e., the matrices in $\mathscr{E}$ are Hermitian, and $E_iE_j$ equals $E_i$ when $i=j$, and equals $O_{p,p}$ otherwise.} for its Bose--Mesner algebra $\mathfrak{B}$. 
\end{fact}
\noindent The interaction between the two bases $\mathscr{S}$ and $\mathscr{E}$ allows deriving several interesting features of association schemes. The change-of-basis matrix shall be particularly important for our purposes. More precisely, we can (uniquely) express the elements of $\mathscr{S}$ as 
\begin{align}
\label{eqn_1819_23_05}
    S_j=\sum_{i=0}^d p_{ij}E_i
\end{align}
for some coefficients $p_{ij}$.
The $(d+1)\times(d+1)$ matrix $P=(p_{ij})$ is known as the \emph{character table} of the association scheme~\cite{bannai_ito_1984}.

For our purposes, association schemes will provide a natural language for describing how the SDA algorithm---in particular, its semidefinite programming part---processes the symmetries of the input digraphs. 
We say that a digraph $\X$ is \emph{generously transitive} if for any $x,y\in \Vset(\X)$ there exists $\xi\in\Aut(\X)$ such that $\xi(x)=y$ and $\xi(y)=x$.
It turns out that the set of orbitals for a generously transitive digraph forms an association scheme.
Indeed, in this case, the condition (\ref{s1}) is trivially satisfied, since each generously transitive digraph is in particular vertex-transitive, while (\ref{s2}) follows from the fact that the orbitals partition $\Vset(\X)^2$.
The conditions (\ref{s3})---in fact, the stronger condition that $R_\omega^T=R_\omega$ for each $\omega\in\Orb(\X)$---and (\ref{s5}) directly come from the definition of generous transitivity. As for the condition (\ref{s4}), it can be proved by considering the Hecke ring of the permutation representation of $\Aut(\X)$; we refer the reader to~\cite{bannai_ito_1984} or~\cite{godsil2016erdos} for further details.
\begin{thm}[\cite{bannai_ito_1984,godsil2016erdos}]
\label{thm_generously_transitive_digraphs_form_association_schemes}
    Let $\X$ be a generously transitive digraph. Then the set $\{R_\omega:\omega\in\Orb(\X)\}$ is a symmetric\footnote{An association scheme $\mathscr{S}$ is \emph{symmetric} if $\mathscr{S}$ consists of symmetric matrices.} association scheme.
\end{thm}
We shall refer to the character table of the association scheme $\{R_\omega:\omega\in\Orb(\X)\}$ as the character table of $\X$. Note that this is a $\lvert\Orb(\X)\rvert\times\lvert\Orb(\X)\rvert$ matrix. Recall that our current objective is to decipher the spectral properties of a balanced matrix $M$ from the corresponding orbital matrix $V$. 
The idea is then to consider a new basis for the space of balanced matrices, alternative to the one of Proposition~\ref{prop_orbital_matrix_of_balanced_matrix}, given by the Kronecker product of the orthogonal bases from Fact~\ref{fact_second_bases_Bose_Mesner} for the two schemes of the input digraphs $\X$ and $\A$.
As the next result shows, working in this new basis allows recovering the spectrum of a balanced matrix from the corresponding orbital matrix. Moreover, the character table serves as the dictionary required for the translation.

\begin{thm}
\label{prop_decomposing_spectrum_M_orbital}
    Let $\X$ and $\A$ be generously transitive digraphs, let $M$ be a balanced matrix for $\X,\A$, let $V$ be the orbital matrix of $M$, and let $P$ and $\tilde{P}$ be the character tables of $\X$ and $\A$, respectively. 
    Then the spectrum of $M$ consists of the entries of the matrix $P V \tilde P^T$.
\end{thm}
\begin{proof}
    Let $\mathscr{E}=\{E_\omega:\omega\in\Orb(\X)\}$ be an orthogonal basis of Hermitian idempotents for the Bose--Mesner algebra of the association scheme corresponding to $\Orb(\X)$, as per Fact~\ref{fact_second_bases_Bose_Mesner}. Also, let $\mathscr{\tilde E}=\{E_{\tilde\omega}:\tilde\omega\in\Orb(\A)\}$ be an orthogonal basis of Hermitian idempotents for the Bose--Mesner algebra for $\Orb(\A)$. Note that $\mathscr{E}$ consists of $p\times p$ matrices, while $\mathscr{\tilde E}$ consists of $n\times n$ matrices (where $p=\lvert\Vset(\X)\rvert$ and $n=\lvert\Vset(\A)\rvert$, as usual).
    Denote the entries of $P$ and $\tilde{P}$ by $p_{ij}$ and $\tilde p_{ij}$, respectively. Using Proposition~\ref{prop_orbital_matrix_of_balanced_matrix}, we can express the balanced matrix $M$ in the basis $\mathscr{R}$ as in~\eqref{eqn_1334_0504}.
    For $\sigma\in\Orb(\X)$ and $\tilde \sigma\in\Orb(\A)$, we find
    \begin{align*}
        M(E_\sigma\otimes E_{\tilde \sigma})
        \spaceeq
        \sum_{\substack{\omega\in\Orb(\X)\\\tilde\omega\in\Orb(\A)}}v_{\omega\tilde\omega}(R_\omega\otimes R_{\tilde\omega})(E_\sigma\otimes E_{\tilde \sigma})
        \spaceeq
        \sum_{\substack{\omega\in\Orb(\X)\\\tilde\omega\in\Orb(\A)}}v_{\omega\tilde\omega}(R_\omega E_\sigma)\otimes (R_{\tilde\omega} E_{\tilde \sigma}).
    \end{align*}
Using~\eqref{eqn_1819_23_05}---and recalling that, by Theorem~\ref{thm_generously_transitive_digraphs_form_association_schemes}, the matrices $R_\omega$ and $R_{\tilde\omega}$ take the roles of the Schur-idempotents for the respective association schemes---we obtain
\begin{align}
\label{eqn_2010_23_05}
\notag
        M(E_\sigma\otimes E_{\tilde \sigma})
        &\spaceeq
        \sum_{\substack{\omega\in\Orb(\X)\\\tilde\omega\in\Orb(\A)}}v_{\omega\tilde\omega}\left(\sum_{\pi\in\Orb(\X)} p_{\pi\omega}E_\pi E_\sigma\right)\otimes \left(\sum_{\tilde\pi\in\Orb(\A)} \tilde p_{\tilde\pi\tilde\omega}E_{\tilde\pi} E_{\tilde\sigma}\right)\\
        \notag
        &\spaceeq
        \sum_{\substack{\omega\in\Orb(\X)\\\tilde\omega\in\Orb(\A)}}v_{\omega\tilde\omega}(p_{\sigma\omega}E_\sigma)\otimes (\tilde p_{\tilde\sigma\tilde\omega}E_{\tilde\sigma})
        \spaceeq
        \sum_{\substack{\omega\in\Orb(\X)\\\tilde\omega\in\Orb(\A)}}v_{\omega\tilde\omega}\,p_{\sigma\omega}\tilde p_{\tilde\sigma\tilde\omega}(E_\sigma\otimes E_{\tilde\sigma})\\
        &\spaceeq
        (\be_{\sigma}^TPV\tilde P^T\be_{\tilde\sigma})(E_\sigma\otimes E_{\tilde\sigma}),
\end{align}
where the second equality follows from the fact that the members of the bases $\mathscr{E}$ and $\mathscr{\tilde E}$ are orthogonal and idempotent.
Consider now, for $\sigma\in\Orb(\X)$ and $\tilde\sigma\in\Orb(\A)$, the complex vector space 
\begin{align*}
    C^{(\sigma,\tilde\sigma)}
    \spaceeq
    (E_\sigma\otimes E_{\tilde\sigma})\mathbb{C}^{pn}
    \spaceeq
    \{(E_\sigma\otimes E_{\tilde\sigma})\bv:\bv\in\mathbb{C}^{pn}\}
    \;\;\subseteq\;\;
    \mathbb{C}^{pn}.
\end{align*}
We claim that
\begin{align}
\label{eqn_2305_2024}
    \mathbb{C}^{pn}
    \spaceeq
    \sum_{\substack{\sigma\in\Orb(\X)\\\tilde\sigma\in\Orb(\A)}} C^{(\sigma,\tilde\sigma)},
\end{align}
where ``$\sum$'' denotes the sum of vector subspaces of $\mathbb{C}^{pn}$.
Using (\ref{s1}), let $\tau\in\Orb(\X)$ and $\tilde\tau\in\Orb(\A)$ be such that $R_\tau=I_p$ and $R_{\tilde\tau}=I_n$. By~\eqref{eqn_1819_23_05}, we have that 
\begin{align*}
    I_p
    \spaceeq
    \sum_{\sigma\in\Orb(\X)}p_{\sigma\tau}E_\sigma
    &&
    \mbox{and}
    &&
    I_n
    \spaceeq
    \sum_{\tilde\sigma\in\Orb(\A)}\tilde p_{\tilde\sigma\tilde\tau}E_{\tilde\sigma}.
\end{align*}
As a consequence, we find that
\begin{align*}
    I_{pn}
    &\spaceeq
    I_p\otimes I_n
    \spaceeq
    \sum_{\substack{\sigma\in\Orb(\X)\\\tilde\sigma\in\Orb(\A)}}p_{\sigma\tau}\tilde p_{\tilde\sigma\tilde\tau}E_\sigma\otimes E_{\tilde\sigma}.
\end{align*}
Given a vector $\bv\in\mathbb{C}^{pn}$, we obtain
\begin{align*}
    \bv
    &\spaceeq
    I_{pn}\bv
    \spaceeq
    \sum_{\substack{\sigma\in\Orb(\X)\\\tilde\sigma\in\Orb(\A)}}p_{\sigma\tau}\tilde p_{\tilde\sigma\tilde\tau}(E_\sigma\otimes E_{\tilde\sigma})\bv
    \;\;\in\;\;
    \sum_{\substack{\sigma\in\Orb(\X)\\\tilde\sigma\in\Orb(\A)}}C^{(\sigma,\tilde\sigma)},
\end{align*}
thus proving the nontrivial inclusion in the claimed identity. It follows from the orthogonality of the idempotents that the sum in~\eqref{eqn_2305_2024} is in fact an orthogonal direct sum.
Furthermore, by~\eqref{eqn_2010_23_05}, each $C^{(\sigma,\tilde\sigma)}$ is an eigenspace for $M$ relative to the eigenvalue $\be_{\sigma}^TPV\tilde P^T\be_{\tilde\sigma}$.
As a consequence,~\eqref{eqn_2305_2024} yields a decomposition of $\mathbb{C}^{pn}$ into eigenspaces\footnote{\label{footnote_complex_vs_real_general_thm}Since the schemes corresponding to $\Orb(\X)$ and $\Orb(\A)$ are symmetric (see~Theorem~\ref{thm_generously_transitive_digraphs_form_association_schemes}), the members of the bases $\mathscr{E}$ and $\mathscr{\tilde E}$ could in fact be chosen to be real, and there would have been no loss of information if we had worked with the real vector spaces $
    (E_\sigma\otimes E_{\tilde\sigma})\mathbb{R}^{pn}$ instead of $
    (E_\sigma\otimes E_{\tilde\sigma})\mathbb{C}^{pn}$.} for $M$, and it follows that the eigenvalues of $M$ are precisely the numbers $\be_{\sigma}^TPV\tilde P^T\be_{\tilde\sigma}$ (i.e., the entries of the matrix $PV\tilde P^T$), with geometric and algebraic multiplicity given by the dimension of $C^{(\sigma,\tilde\sigma)}$. 
\end{proof}

One consequence of Theorem~\ref{prop_decomposing_spectrum_M_orbital} is that, in the new basis for the space of balanced matrices, the \emph{semidefinite-programming} condition $M\succcurlyeq 0$ is transformed into the \emph{linear-programming} condition $PV\tilde P^T\geq 0$. All other conditions making $M$ an $\SDP$-matrix (namely, the conditions (\ref{r1})--(\ref{r5}) and the entrywise nonnegativity) are trivially translated into equivalent (linear-programming) conditions on $V$. Hence, Theorem~\ref{prop_decomposing_spectrum_M_orbital} turns the semidefinite program~\eqref{eqns_SDP} applied to two generously transitive digraphs $\X$ and $\A$ into an equivalent linear program, whose constraints are now in terms of the character tables of $\X$ and $\A$. This is made explicit in the next result.

For a digraph $\X$, we let $\bmu^\X$ be the vector, indexed by the elements of $\Orb(\X)$, whose $\omega$-th entry is $\lvert\omega\rvert$. 
We say that an orbital $\omega$ is the \emph{diagonal orbital} if $R_\omega$ is the identity matrix, and we say that $\omega$ is an \emph{edge orbital} if $\omega\subseteq\Eset(\X)$; \emph{non-diagonal} and \emph{non-edge} orbitals are defined in the obvious way. Notice that the edge orbitals of $\X$ partition $\Eset(\X)$.

\begin{cor}
\label{cor_critierion_acceptance_SDP}
    Let $\X$ and $\A$ be generously transitive digraphs and let $P$ and $\tilde{P}$ be the character tables of $\X$ and $\A$, respectively. Then $\SDP(\X,\A)=\YES$ if and only if there exists a real entrywise-nonnegative $\lvert\Orb(\X)\rvert\times\lvert\Orb(\A)\rvert$ matrix $V$ such that
    \begin{multicols}{2}
    \begin{itemize}
        \item[$(c_1)$\labeltext{$c_1$}{c1}] $P V \tilde P^T\geq 0$;
        \item[$(c_2)$\labeltext{$c_2$}{c2}] $V\bmu^\A=\bone$;
    \end{itemize}
    \end{multicols}
    \vspace{-.5cm}
    \begin{itemize}
        \item[$(c_3)$\labeltext{$c_3$}{c3}] $v_{\omega\tilde\omega}=0$ if $\omega$ is the diagonal orbital of $\X$ and $\tilde\omega$ is a non-diagonal orbital of $\A$;
        \item[$(c_4)$\labeltext{$c_4$}{c4}] $v_{\omega\tilde\omega}=0$
        if $\omega$ is an edge orbital of $\X$ and $\tilde\omega$ is a non-edge orbital of $\A$.
    \end{itemize} 
\end{cor}
\begin{proof}
    Combining Proposition~\ref{prop_matrix_char_SDP_new_form} with Proposition~\ref{prop_orbital_matrix_of_balanced_matrix}, we find that $\SDP(\X,\A)=\YES$ if and only if there exists a real $\lvert\Orb(\X)\rvert\times\lvert\Orb(\A)\rvert$ matrix $V$ such that, letting
     \begin{align}       
     \label{eqn_25_05_1039}
     M
     \spaceeq
     \sum_{\substack{\sigma\in\Orb(\X)\\\tilde\sigma\in\Orb(\A)}}v_{\sigma\tilde\sigma}\,R_\sigma\otimes R_{\tilde\sigma},
    \end{align}
    $M$ is an $\SDP$-matrix. Hence, we need to show that $M$ is an $\SDP$-matrix precisely when $V$ is entrywise nonnegative and satisfies (\ref{c1})--(\ref{c4}). 

One readily sees from~\eqref{eqn_25_05_1039} that $M\geq 0$ exactly when $V\geq 0$. Since the association schemes corresponding to $\Orb(\X)$ and $\Orb(\A)$ are symmetric by Theorem~\ref{thm_generously_transitive_digraphs_form_association_schemes}, all matrices $R_\sigma$ and $R_{\tilde\sigma}$ are symmetric, and so are their Kronecker products. Hence, by~\eqref{eqn_25_05_1039}, $M$ is symmetric as well.\footnote{Note that $V$ is not required to be symmetric; in fact, $V$ is not square in general.} It follows that $M\succcurlyeq 0$ if and only if its spectrum is nonnegative. By Theorem~\ref{prop_decomposing_spectrum_M_orbital}, this is equivalent to (\ref{c1}).
Consider two orbitals $\omega\in\Orb(\X)$ and $\tilde\omega\in\Orb(\A)$, and let $(x,y)\in\omega$, $(a,b)\in\tilde\omega$. We find
    \begin{align*}
        (\be_x\otimes\be_a)^TM(\be_y\otimes\be_b)
        &\spaceeq
        \sum_{\substack{\sigma\in\Orb(\X)\\\tilde\sigma\in\Orb(\A)}}v_{\sigma\tilde\sigma}(\be_x\otimes\be_a)^T(R_\sigma\otimes R_{\tilde\sigma})(\be_y\otimes\be_b)\\
        &\spaceeq
        \sum_{\substack{\sigma\in\Orb(\X)\\\tilde\sigma\in\Orb(\A)}}v_{\sigma\tilde\sigma}(\be_x^T R_\sigma\be_y)(\be_a^TR_{\tilde\sigma}\be_b)
        \spaceeq
        v_{\omega\tilde\omega}.
    \end{align*}
It follows that (\ref{c3}) and (\ref{c4}) are equivalent to (\ref{r1}) and (\ref{r2}), respectively. Furthermore, given $(x,y)\in\omega\in\Orb(\X)$, we have
\begin{align*}
    (e_x\otimes\bone_n)^TM(\be_y\otimes\bone_n)
    &\spaceeq
    \sum_{(a,b)\in \Vset(\A)^2}(\be_x\otimes\be_a)^TM(\be_y\otimes\be_b)\\
    &\spaceeq
    \sum_{\tilde\omega\in\Orb(\A)}\sum_{(a,b)\in\tilde\omega}(\be_x\otimes\be_a)^TM(\be_y\otimes\be_b)
    \spaceeq
    \sum_{\tilde\omega\in\Orb(\A)}\sum_{(a,b)\in\tilde\omega}v_{\omega\tilde\omega}\\
    &\spaceeq
    \sum_{\tilde\omega\in\Orb(\A)}v_{\omega\tilde\omega}\lvert\tilde\omega\rvert
    \spaceeq
    \be_\omega^TV\bmu^\A.
\end{align*} 
As a consequence, (\ref{c2}) is equivalent to (\ref{r6})---which, by Proposition~\ref{prop_equivalence_some_conditions_SDP}, is equivalent to (\ref{r3}), (\ref{r4}), and (\ref{r5}). 
\end{proof}
In order to prove that $\SDA$ does not solve $\PCSP(\A,\B)$ for any pair of non-bipartite loopless undirected graphs such that $\A\to\B$, thus establishing Theorem~\ref{thm_main_SDA_no_solves_approximate_homomorphism}, 
we seek a \emph{fooling instance}: a digraph $\X$ such that $\SDA(\X,\A)=\YES$ but $\X\not\to\B$.
If we wish to apply Corollary~\ref{cor_critierion_acceptance_SDP} and take advantage of the machinery developed so far for describing the output of $\SDP$, we need both $\X$ and $\A$ to be generously transitive digraphs.
Regarding $\A$, this requirement does not create problems. Indeed, it is not hard to check that
it is enough to establish the result in the case that $\A$ is an odd undirected cycle and $\B$ is a clique. Since cycles happen to be generously transitive, Theorem~\ref{thm_generously_transitive_digraphs_form_association_schemes} does apply; as we shall see, the structure of the scheme for odd cycles also allows dealing with the linear part of SDA. The more challenging part is to come up with a digraph $\X$ that $(i)$ is generously transitive, $(ii)$ is not homomorphic to $\B$ (i.e., is highly chromatic), and $(iii)$ is accepted by $\SDA$. A promising candidate is the class of Kneser graphs, as they $(i)$ are generously transitive and $(ii)$ have unbounded chromatic number (that is easily derived from the parameters of the graphs through a classic result by Lov\'{a}sz~\cite{lovasz1978kneser}).
In the next two sections, we look at the association schemes for Kneser graphs and odd cycles. The task is to collect enough information on their character tables to design an orbital matrix witnessing the fact that $(iii)$ $\SDA(\X,\A)=\YES$.

\section{The Johnson scheme}
\label{sec_johnson_scheme}
Given $s,t\in\N$ such that $s>2t$, the \emph{Kneser graph} $\GG_{s,t}$ is the undirected graph whose vertices are all subsets of $[s]$ of size $t$, and whose edges are all disjoint pairs of such subsets. 
As a consequence of the Erd\H{o}s--Ko--Rado theorem~\cite{bollobas1986combinatorics,godsil2016erdos},
the automorphism group of $\GG_{s,t}$ is isomorphic to the symmetric group $\operatorname{Sym}_s$ consisting of the permutations of $[s]$. 
More precisely, given $f\in\operatorname{Sym}_s$, the corresponding automorphism $\xi_f$ of $\GG_{s,t}$ is given by $\xi_f:\{a_1,\dots,a_t\}\mapsto\{f(a_1),\dots,f(a_t)\}$ for any set $\{a_1,\dots,a_t\}$ of $t$ elements of $[s]$.  
Let $U,V$ be vertices of $\GG_{s,t}$, let $Z=U\cap V$, and label the elements of $U\setminus Z$, $V\setminus Z$, and $Z$ by $\{u_1,\dots,u_{q}\}$, $\{v_1,\dots,v_{q}\}$, and $\{z_1,\dots,z_{t-q}\}$ for some $0\leq q\leq t$. Letting $f\in\operatorname{Sym}_s$ be the permutation that switches each $u_i$ with $v_i$ and is constant over $[s]\setminus(U\cup V\setminus Z)$, we see that $\xi_f(U)=V$ and $\xi_f(V)=U$. This means that $\GG_{s,t}$ is generously transitive and, thus, $\Orb(\GG_{s,t})$ generates an association scheme, which is known as the \emph{Johnson scheme}.

Observe that, for two vertices $U$ and $V$ as above, the orbital of $(U,V)$ contains precisely all pairs of vertices $(U',V')$ such that $\lvert U'\cap V'\rvert=\lvert U\cap V\rvert=t-q$.
Hence, the association scheme corresponding to $\Orb(\GG_{s,t})$ consists of the adjacency matrices of the \emph{generalised Johnson graphs} $\JJ_{s,t,q}$ for $q=0,\dots,t$, where $\JJ_{s,t,q}$ is the graph having the same vertex set as $\GG_{s,t}$, with two vertices being adjacent if and only if their intersection has size $t-q$ (cf.~Figure~\ref{fig_JohnsonGraphs}). 
For $q=t$, $\JJ_{s,t,q}$ is $\GG_{s,t}$; for $q=1$, it is known as the \emph{Johnson graph}, see~\cite[\S~1.6]{godsil2001algebraic}; for $q=0$, it is the disjoint union of $\bin{s}{t}$ loops (whose adjacency matrix is the identity). 
Hence, the diagonal orbital corresponds to $q=0$, while the (unique) edge orbital corresponds to $q=t$.

\begin{figure}
\centering
\hspace{-2.7cm}
\begin{subfigure}{.16\textwidth}
\adjincludegraphics[width=2\linewidth,trim={{.35\width} {.5\width} {.33\width} {.5\width}},clip]{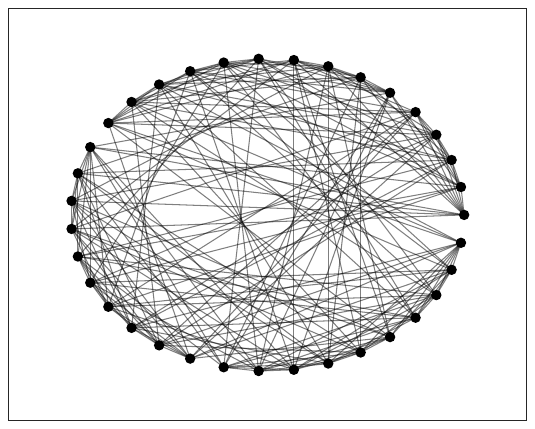}
\end{subfigure}%
\hspace{3cm}
\begin{subfigure}{.16\textwidth}
\adjincludegraphics[width=2\linewidth,trim={{.35\width} {.5\width} {.33\width} {.5\width}},clip]{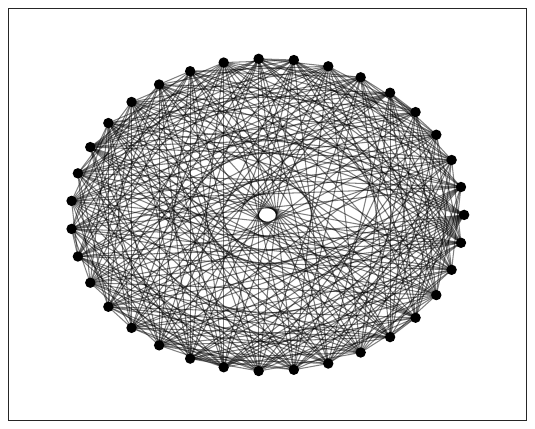}
\end{subfigure}%
\hspace{3cm}
\begin{subfigure}{.16\textwidth}
\adjincludegraphics[width=2\linewidth,trim={{.35\width} {.5\width} {.33\width} {.5\width}},clip]{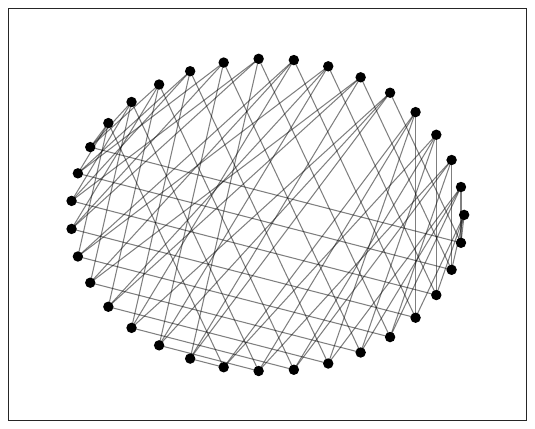}
\end{subfigure}%
\caption{From the left, the generalised Johnson graphs $\JJ_{7,3,1}$, $\JJ_{7,3,2}$, and $\JJ_{7,3,3}$.}
\label{fig_JohnsonGraphs}
\end{figure}

In order to design an orbital matrix witnessing that $\GG_{s,t}$ is accepted by $\SDP$ (and, as we will see, $\SDA$), it shall be useful to gain some insight into the behaviour of the character table of $\GG_{s,t}$ when it is multiplied by column vectors (which, ultimately, will be the columns of the orbital matrix, cf.~Corollary~\ref{cor_critierion_acceptance_SDP}).
We shall see that, if a column vector is interpolated by a polynomial of low degree, multiplying it by the character table yields a vector living in a fixed, low-dimensional subspace of $\R^{t+1}$. This observation leads us to choose
an orbital matrix 
whose nonzero columns are polynomials of degree one (cf.~the proof of Theorem~\ref{thm_main_SDA_no_solves_approximate_homomorphism} in Section~\ref{subsec_lower_bound_on_SDA}). 

Let $\bh$ be the vector $(0,1,\dots,t)$ and, given a univariate polynomial $f\in\R[x]$, let $\bh^f$ be the vector $(f(0),f(1),\dots,f(t))$.
Henceforth, we shall label
the members of the Johnson scheme by $0,1,\dots,t$, where the $q$-th member is $\adj(\JJ_{s,t,q})$. Hence, we label the entries of the character table of $\GG_{s,t}$ and the standard unit vectors $\be_i$ in Theorem~\ref{thm_description_spectral_matrix_Kneser_times_polynomial} accordingly, with indices ranging over $\{0,\dots,t\}$ rather than $\{1,\dots,t+1\}$. A similar labelling shall also be used for the cycle scheme, cf.~Propositions~\ref{prop_description_association_scheme_cycle} and~\ref{prop_fooling_matrices_cycle}.
\begin{thm}
\label{thm_description_spectral_matrix_Kneser_times_polynomial}
Let $s,t\in\N$ with $s>2t$, and let $P$ be the character table of $\GG_{s,t}$. Then
\begin{itemize}
    \item[(i)] $P \bh^f\in\spann(\be_0,\dots,\be_d)$ for any univariate polynomial $f$ of degree $d\leq t$;
\end{itemize}
\vspace{-.2cm}
\begin{multicols}{2}
    \begin{itemize}
        \item[(ii)]$P\bone=\bin{s}{t}\be_0$;
        \item[(iii)]$P\bh=\bin{s}{t}(\frac{st-t^2}{s}\be_0-\frac{st-t^2}{s^2-s}\be_1)$.  
    \end{itemize}
\end{multicols}
\end{thm}
To prove Theorem~\ref{thm_description_spectral_matrix_Kneser_times_polynomial}, we take advantage of the explicit expression for the character table of the Johnson scheme obtained by Delsarte~\cite{delsarte1973algebraic} (see also~\cite[\S~6.5]{godsil2016erdos}) in terms of the \emph{Eberlein polynomials} defined,
for $s,t,q,j\in\N_0$, by\footnote{This expression for $\beta(s,t,q,j)$ is obtained from the one given in~\cite[\S~6.5]{godsil2016erdos} via straightforward manipulations.} 
\begin{align*}
        \beta(s,t,q,j)
    \spaceeq
    \sum_{i=0}^\infty(-1)^{i-q+j}\bin{i}{q}\bin{t-j}{i-j}\bin{s-i-j}{t-j}.
\end{align*}
Here, we are using the conventions that $\bin{x}{y}=0$ unless $0\leq y\leq x$, and that $\bin{0}{0}=1$. In particular, this implies that $\beta(s,t,q,j)=0$ unless $q,j\leq t\leq s$.
\begin{thm}[\cite{delsarte1973algebraic}]
\label{thm_entries_spectral_matrix_Kneser}
    Let $s,t\in\N$ be such that $s>2t$. Then the character table of the Kneser graph $\GG_{s,t}$ is the $(t+1)\times(t+1)$ matrix whose $(j,q)$-th entry, for $j,q\in\{0,\dots,t\}$, is $\beta(s,t,q,j)$.
\end{thm}
Our strategy consists in associating with the entries of the character table a family of bivariate generating functions (parameterised by $t$ and $j$) defined by
\begin{align}
\label{eqn_defn_gamma_function}
    \gamma_{t,j}(x,y)
    \spaceeq
    \sum_{s,q\in\N_0}\beta(s,t,q,j) x^s y^q.
\end{align}
We now find a closed formula for these generating functions. Henceforth, the range in a summation shall always be meant to be $\N_0$ unless otherwise specified.

\begin{prop}
\label{prop_generating_function_eigenvalues_Kneser}
The identity
    \begin{align*}
        \gamma_{t,j}(x,y)
        \spaceeq
        x^{t+j}(1-x)^{j-t-1}(1-y)^j(1-x+xy)^{t-j}
    \end{align*}
    holds for each  $t,j\in\N_0$ and $x,y\in\R$ such that $j\leq t$ and $-1<x<1$.
\end{prop}
\begin{proof}
We use the well-known identity
    \begin{align*}
        \sum_{n}\bin{n}{a}x^n\spaceeq\frac{x^a}{(1-x)^{a+1}},
    \end{align*}
    valid for $a\in\N_0$ in the disk $-1<x<1$ (see~\cite[\S~5.4]{Graham94:concrete}).
    In particular, this implies, for $b\in\N_0$,
    \begin{align}
    \label{eqn_1904_1210}
        \sum_{n} \bin{n-b}{a}x^n
        &\spaceeq
        \sum_{n}\bin{n}{a}x^{n+b}
        \spaceeq
        x^b \sum_{n}\bin{n}{a}x^n
        \spaceeq
        \frac{x^{a+b}}{(1-x)^{a+1}}.
    \end{align}
We find
{\allowdisplaybreaks
    \begin{align*}
        \gamma_{t,j}(x,y)
        &\spaceeq
        \sum_{s,q,i} (-1)^{i-q+j}\bin{i}{q}\bin{t-j}{i-j}\bin{s-i-j}{t-j}x^sy^q\\
        &\spaceeq
        \sum_{q,i}(-1)^{i-q+j}\bin{i}{q}\bin{t-j}{i-j}y^q\sum_s \bin{s-i-j}{t-j}x^s\\
        &\equationeq{eqn_1904_1210}
        \sum_{q,i}(-1)^{i-q+j}\bin{i}{q}\bin{t-j}{i-j}y^qx^{t+i}(1-x)^{j-t-1}\\
        &\spaceeq
        (1-x)^{j-t-1}\sum_i\bin{t-j}{i-j}x^{t+i}\sum_q (-1)^{i-q+j}\bin{i}{q}y^q\\
        &\spaceeq
        (1-x)^{j-t-1}\sum_i\bin{t-j}{i-j}x^{t+i}(-1)^j(y-1)^i\\
        &\spaceeq
        (1-x)^{j-t-1}\sum_i\bin{t-j}{i-j}x^{t+i}(-1)^{i-j}(1-y)^i\\
        &\spaceeq
        x^{t+j}(1-x)^{j-t-1}(1-y)^j\sum_i\bin{t-j}{i-j}x^{i-j}(-1)^{i-j}(1-y)^{i-j}\\
        &\spaceeq
        x^{t+j}(1-x)^{j-t-1}(1-y)^j\sum_i\bin{t-j}{i-j}(xy-x)^{i-j}\\
        &\spaceeq
        x^{t+j}(1-x)^{j-t-1}(1-y)^j\sum_i\bin{t-j}{i}(xy-x)^{i}\\
        &\spaceeq
        x^{t+j}(1-x)^{j-t-1}(1-y)^j(1-x+xy)^{t-j},
    \end{align*}
}
    as required.
\end{proof}

Theorem~\ref{thm_description_spectral_matrix_Kneser_times_polynomial} is then proved by expressing the entries of the vector $P\bh^f$ in terms of partial derivatives of the generating functions $\gamma_{t,j}$, and by finding analytic expressions for these partial derivatives through Proposition~\ref{prop_generating_function_eigenvalues_Kneser}. 
In particular, the quantity
\begin{align*}
    \vartheta_{s,t,j}(k)
    \spaceeq
    \sum_{q}q^k\beta(s,t,q,j)
\end{align*}
shall be crucial in the following. Observe that we can view it as the $j$-th entry of the vector obtained by multiplying the character table of $\GG_{s,t}$ by a vector whose $q$-th entry is $q^k$.
It is possible to isolate $\vartheta_{s,t,j}(k)$ by differentiating the generating functions $\gamma_{t,j}$. Using the closed formula for $\gamma_{t,j}$ we just obtained, one can then deduce some useful identities.

\begin{prop}
\label{prop_P_kneser_polynomial_is_triangular}
Let $s,t,q,j,k\in\N_0$. Then
\begin{itemize}
    \item[(i)] $\vartheta_{s,t,j}(k)=0$ if $k<j$; 
    \item[(ii)]
    $\vartheta_{s,t,j}(j)=(-1)^jj!\bin{s-2j}{t-j}$;
    \item[(iii)]
    $\vartheta_{s,t,j}(j+1)=(-1)^j(j+1)!\bin{s-2j}{t-j}\left(t-\frac{j}{2}-\frac{(t-j)^2}{s-2j}\right)$.
\end{itemize}
\end{prop}
\begin{proof}
Notice that the result is trivially true if $j>t$, as in this case
$\beta(s,t,q,j)=0$.
Differentiating the polynomial $\gamma_{t,j}$ as defined in~\eqref{eqn_defn_gamma_function} $k$ times with respect to the variable $y$ yields
\begin{align}
\label{eqn_1719_1904}
    \frac{\partial^k\gamma_{t,j}}{\partial y^k}
        &\spaceeq
        \sum_{s,q}\beta(s,t,q,j)x^sk!\bin{q}{k}y^{q-k}.
\end{align}
Observe that $k!\bin{q}{k}$ is a polynomial in $q$ of degree $k$. Hence, we can find coefficients $a^{(k)}_0,\dots,a^{(k)}_k$ for which $k!\bin{q}{k}=\sum_{i=0}^k a^{(k)}_iq^i$; in particular, we have $a^{(k)}_k=1$.
Evaluating~\eqref{eqn_1719_1904} at $y=1$ yields
\begin{align}
\label{eqn_1549_1904_a}
\notag
    \left.\frac{\partial^k\gamma_{t,j}}{\partial y^k}\right\vert_{y=1}
        &\spaceeq
        \sum_{s,q}\beta(s,t,q,j)x^sk!\bin{q}{k}
        \spaceeq
        \sum_{s,q}\beta(s,t,q,j)x^s\sum_{i=0}^k a^{(k)}_iq^i\\
        &\spaceeq
        \sum_{s}x^s\sum_{i=0}^k a^{(k)}_i \sum_q q^i\beta(s,t,q,j)
        \spaceeq
        \sum_{s}x^s\sum_{i=0}^k a^{(k)}_i \vartheta_{s,t,j}(i).
\end{align}
We now make use of Proposition~\ref{prop_generating_function_eigenvalues_Kneser} to get an alternative expression for the object above. If $-1<x<1$, applying the Leibniz rule for differentiation, we obtain
    \begin{align*}
        \frac{\partial^k\gamma_{t,j}}{\partial y^k}
        &\spaceeq
        x^{t+j}(1-x)^{j-t-1}\sum_{\ell=0}^k\bin{k}{\ell}\frac{\partial^{\ell}((1-y)^j)}{\partial y^{\ell}}\;\frac{\partial^{k-\ell}((1-x+xy)^{t-j})}{\partial y^{k-\ell}}.
    \end{align*}
    Notice that, for $0\leq\ell\leq k$, 
    \begin{align*}
        \left.\frac{\partial^{\ell}((1-y)^j)}{\partial y^\ell}\right\vert_{y=1}
        &\spaceeq
        \left\{
        \begin{array}{llll}
             (-1)^jj!&\mbox{ if }\ell=j  \\
             0&\mbox{ otherwise} 
        \end{array}
        \right.
        \quad\quad\quad\mbox{and}\\
        \left.\frac{\partial^{k-\ell}((1-x+xy)^{t-j})}{\partial y^{k-\ell}}\right\vert_{y=1}
        &\spaceeq
        (k-\ell)!\bin{t-j}{k-\ell}x^{k-\ell}.
    \end{align*}
    If $k\geq j$,
    it follows that, over the disk $-1<x<1$,
    \begin{align}
    \label{eqn_2004_1712}
    \notag
        \left.\frac{\partial^k\gamma_{t,j}}{\partial y^k}\right\vert_{y=1}
        &\spaceeq
        x^{t+j}(1-x)^{j-t-1}\bin{k}{j}(-1)^jj!(k-j)!\bin{t-j}{k-j}x^{k-j}\\
        \notag
        &\spaceeq
        x^{t+k}(1-x)^{j-t-1}k!(-1)^j\bin{t-j}{k-j}\\
        \notag
        &\spaceeq
        \sum_\ell (-1)^j k!\bin{t-j}{k-j}\bin{t-j+\ell}{\ell}x^{\ell+t+k}\\
        &\spaceeq
        \sum_s x^s(-1)^j k!\bin{t-j}{k-j}\bin{s-j-k}{t-j}.
    \end{align}
In fact,~\eqref{eqn_2004_1712} also holds if $k<j$, as both terms in the equality are zero in that case.
Using that $\gamma_{t,j}$ is an analytic function, we can then compare~\eqref{eqn_1549_1904_a} and~\eqref{eqn_2004_1712} equating the coefficients. This yields the identity
\begin{align}
\label{eqn_2004_1804}
    \sum_{i=0}^k a^{(k)}_i \vartheta_{s,t,j}(i)
    \spaceeq
    (-1)^j k!\bin{t-j}{k-j}\bin{s-j-k}{t-j}.
\end{align}

If $k<j$, the right-hand side of~\eqref{eqn_2004_1804} is zero.
Recalling that $a^{(k)}_k=1$,
    it is then immediate to conclude  by induction over $k$ that $\vartheta_{s,t,j}(k)=0$, thus establishing $(i)$.

If $k=j$, using $(i)$ we find that the left-hand side of~\eqref{eqn_2004_1804} is $\vartheta_{s,t,j}(j)$, while the right-hand side is $(-1)^jj!\bin{s-2j}{t-j}$, thus establishing $(ii)$.

If $k=j+1$, using $(i)$, $(ii)$, and the fact that $a^{(j+1)}_{j}
        =
        -\frac{j^2+j}{2}$, we find that the left-hand side of~\eqref{eqn_2004_1804} is
\begin{align*}
    a^{(j+1)}_j\vartheta_{s,t,j}(j)+a^{(j+1)}_{j+1}\vartheta_{s,t,j}(j+1)
    &\spaceeq
    -\frac{j^2+j}{2}(-1)^jj!\bin{s-2j}{t-j}+\vartheta_{s,t,j}(j+1),
\end{align*}
while the right-hand side is
\begin{align*}
    (-1)^j(j+1)!(t-j)\bin{s-2j-1}{t-j}
    \spaceeq
    (-1)^j(j+1)!\bin{s-2j}{t-j}\left(t-j-\frac{(t-j)^2}{s-2j}\right).
\end{align*}
Equating the two sides gives the expression in $(iii)$ for $\vartheta_{s,t,j}(j+1)$.
\end{proof}

\begin{rem}
    We observe that the equation~\eqref{eqn_2004_1804} yields the following recursive identity satisfied by the quantities $\vartheta_{s,t,j}(k)$:
    \begin{align*}
        \vartheta_{s,t,j}(k)
        \spaceeq
        -\sum_{i=j}^{k-1}a^{(k)}_i\vartheta_{s,t,j}(i)+(-1)^j k!\bin{t-j}{k-j}\bin{s-j-k}{t-j}.
    \end{align*}
    The numbers $a^{(k)}_i$ are the Stirling numbers of the first kind, see~\cite{MR2868112}.
\end{rem}

We can now prove Theorem~\ref{thm_description_spectral_matrix_Kneser_times_polynomial}.
Recall that $\bh=(0,1,\dots,t)$ and $\bh^f=(f(0),f(1),\dots,f(t))$ for a univariate polynomial $f\in\R[x]$ (where, as usual, we interpret tuples as column vectors); if $f$ is the monomial given by $f(x)=x^k$, we denote $\bh^f$ by $\bh^k$. 
We now establish that multiplying the character table of the Johnson scheme by $\bh^f$ yields a vector with the property that the entries with index bigger than the degree of $f$ are zero.\footnote{In accordance with the labelling of the orbitals of $\GG_{s,t}$ and of the entries of the character table, the entries of vectors of length $t+1$ shall have indices ranging in $\{0,\dots,t\}$.} 
This fact is particularly useful when $f$ has a low degree, as in this case all but the first few entries in the resulting vector are zero. Using parts $(ii)$ and $(iii)$ of Proposition~\ref{prop_P_kneser_polynomial_is_triangular}, we are able to find the expressions for the nonzero coefficients in the case that the degree is zero or one.

\begin{proof}[Proof of Theorem~\ref{thm_description_spectral_matrix_Kneser_times_polynomial}]
Let $f$ be a polynomial of degree $d\leq t$, and write it as $f(x)=\sum_{k=0}^da_kx^k$ for some coefficients $a_k$. We obtain $\bh^f=\sum_{k=0}^d a_k\bh^k$. Hence,
for each $j\in\{0,\dots,t\}$, we have
    \begin{align}
    \label{eqn_2004_1036}
    \notag
        \be_j^T P \bh^f
        &\spaceeq
        \sum_{k=0}^da_k\,\be_j^TP\bh^k
        \spaceeq
        \sum_{k=0}^d a_k\sum_{q=0}^tq^k\be_j^TP\be_q
        \spaceeq
        \sum_{k=0}^d a_k\sum_{q=0}^tq^k\beta(s,t,q,j)\\
        &\spaceeq
        \sum_{k=0}^d a_k\,\vartheta_{s,t,j}(k).
    \end{align}
If $j>d$, it follows from Proposition~\ref{prop_P_kneser_polynomial_is_triangular}$(i)$ that the quantity in~\eqref{eqn_2004_1036} is zero, which means that $P\bh^f\in\spann(\be_0,\dots,\be_d)$, thus proving $(i)$.

Note that $\bone=\bh^0$, so by $(i)$ $P\bone$ is a scalar multiple of $\be_0$. From~\eqref{eqn_2004_1036}, we find that 
\begin{align*}
    \be_0^T P\bone
    \spaceeq
    \vartheta_{s,t,0}(0)
    \propparteq{prop_P_kneser_polynomial_is_triangular}{(ii)}
    \bin{s}{t},
\end{align*}
thus proving $(ii)$.

Similarly, since $\bh=\bh^1$,
by $(i)$ we have that $P\bh\in\spann(\be_0,\be_1)$; the coefficients are given by
\begin{align*}
    \be_0^TP\bh
    &\equationeq{eqn_2004_1036}
    \vartheta_{s,t,0}(1)
    \propparteq{prop_P_kneser_polynomial_is_triangular}{(iii)}
    \bin{s}{t}\left(t-\frac{t^2}{s}\right)
    \quad\quad\quad\mbox{ and }\\
    \be_1^TP\bh
    &\equationeq{eqn_2004_1036}
    \vartheta_{s,t,1}(1)
    \propparteq{prop_P_kneser_polynomial_is_triangular}{(ii)}
    -\bin{s-2}{t-1}
    \spaceeq
    -\frac{st-t^2}{s^2-s}\bin{s}{t}
    ,
\end{align*}
which proves $(iii)$.
\end{proof}

\section{The cycle scheme}
\label{sec_cycle_scheme}
Consider the undirected cycle $\C_n$ with $n$ vertices, where $n\geq 3$.
It is well known that $\Aut(\C_n)$ is the dihedral group of order $2n$, consisting of all rotations and reflections of the cycle. Any pair of distinct vertices is switched by a suitable reflection, so $\C_n$ is generously transitive. Hence, the orbitals of $\C_n$ form an association scheme.
The automorphism group of any graph $\X$ acts isometrically on $\Vset(\X)^2$, in the sense that $\dist(\xi(x),\xi(y))=\dist(x,y)$ for any $x,y\in\Vset(\X)$ and any $\xi\in\Aut(\X)$.
In addition, the structure of the dihedral group implies that two pairs $(x,y)$ and $(x',y')$ of vertices of $\C_n$ lie in the same orbital whenever $\dist(x,y)=\dist(x',y')$.
Hence, if $n=2m+1$ is an odd integer, there are exactly $m+1$ orbitals---one for each possible distance between two vertices in the cycle. We can thus write $\Orb(\C_n)=\{\omega_0,\dots,\omega_m\}$, with $\omega_j=\{(x,y)\in \Vset(\C_n)^2:\dist(x,y)=j\}$. Note that $R_{\omega_0}=I_n$ and $R_{\omega_1}=\adj(\C_n)$. In other words, $\omega_0$ and $\omega_1$ are the diagonal orbital and the (unique) edge orbital, respectively. 
Each orbital has size $2n$ except $\omega_0$, which has size $n$.
Instead of providing a complete description of the character table of $\C_n$, it shall be enough for our purposes to derive one property using the Perron--Frobenius theorem. We say that a real entrywise-nonnegative square matrix $M$ is \emph{primitive} if $M^c$ is entrywise positive for some power $c\in\N$.
\begin{thm}[Perron--Frobenius theorem~\cite{perron1907,frobenius1912}]
\label{thm_perron_frobenius}
Let $M$ be a primitive matrix. Then $M$ has a unique eigenvalue $\rho$ associated with an entrywise-nonnegative eigenvector. Moreover, $\rho$ is a simple eigenvalue\footnote{I.e., it has algebraic multiplicity one.}, it is real and positive, and $\lvert\lambda\rvert<\rho$ for each other eigenvalue $\lambda$ of $M$. 
\end{thm}

\noindent Like in the case of the Johnson scheme, it shall be convenient to let the indices of the entries in the character table of $\C_n$ range over $\{0,\dots,m\}$ rather than $\{1,\dots,m+1\}$.

\begin{prop}
\label{prop_description_association_scheme_cycle}
    Let $n\geq 3$ be an odd integer,
    and let $P$ be the character table of $\C_{n}$. Then $P\be_0=\bone$, while $P\be_1$ contains exactly one entry equal to $2$, and all other entries are strictly smaller than $2$ in absolute value.
\end{prop}
\begin{proof}
    Let $\mathscr{E}=\{E_0,\dots,E_m\}$ be an orthogonal basis of idempotents for the association scheme of $\Orb(\C_n)$ as per Fact~\ref{fact_second_bases_Bose_Mesner}, where $m=\frac{n-1}{2}$. Since $\mathscr{E}$ is a basis, all of its members are different from the zero matrix, so for each $\ell\in\{0,\dots,m\}$ there exists $\bv^{(\ell)}\in\mathbb{C}^n$ such that $E_\ell\bv^{(\ell)}\neq\bzero$. Using the idempotency of the basis, we find, for each $j\in\{0,\dots,m\}$,
    \begin{align}
    \label{eqn_2905_1243}
        R_{\omega_j}E_\ell\bv^{(\ell)}
        \spaceeq
        \sum_{i=0}^m p_{i j }E_i E_\ell\bv^{(\ell)}
        \spaceeq
        p_{\ell j}E_\ell\bv^{(\ell)}.
    \end{align}
This means that $p_{\ell j}$ is an eigenvalue of $R_{\omega_j}$ for each $\ell,j$. Choosing $j=0$ and recalling that $R_{\omega_0}=I_n$, it follows that $P\be_0=\bone$. Reasoning as in the proof of Theorem~\ref{prop_decomposing_spectrum_M_orbital}, we find that $\mathbb{C}^n=\sum_{\ell=0}^m E_\ell\mathbb{C}^n$, which, by~\eqref{eqn_2905_1243},
yields a simultaneous decomposition of $\mathbb{C}^n$ into eigenspaces\footnote{Using that the scheme $\Orb(\C_n)$ is symmetric, we could as well have considered a real eigenspace decomposition, cf.~Footnote~\ref{footnote_complex_vs_real_general_thm}.} of $R_{\omega_j}$ for each $\omega_j\in\Orb(\C_n)$. As a consequence, all eigenvalues of $R_{\omega_j}$ appear as entries of the vector $P\be_j$. 
Since $n$ is odd, given two vertices $x$ and $y$ of $\C_n$, we can always find a path connecting $x$ to $y$ whose length is even; clearly, any such path can be extended to a walk of length exactly $n-1$ connecting $x$ to $y$. This means that the matrix $(\adj(\C_n))^{n-1}$ is entrywise positive, thus showing that $R_{\omega_1}=\adj(\C_n)$ is a primitive matrix. Note that $R_{\omega_1}\bone=2\cdot\bone$. By Theorem~\ref{thm_perron_frobenius}, $\rho=2$ is the spectral radius of $R_{\omega_1}$, it is a simple eigenvalue, and all other eigenvalues of $R_{\omega_1}$ have strictly smaller absolute value. This yields the desired description for $P\be_1$.
\end{proof}

\begin{rem}
    The orbitals of $\C_n$ form an association scheme also in the case that $n$ is even. However, the description of $P\be_1$ in Proposition~\ref{prop_description_association_scheme_cycle} fails to be true in that case, the reason being that the adjacency matrix of an even cycle is not primitive. In fact, it is well known that the spectrum of the adjacency matrix of a bipartite graph is symmetric around $0$ (see~\cite[Prop.~8.2]{MR1271140}). As a consequence, there are two entries of $P\be_1$ whose values are $2$ and $-2$. Ultimately, this slight difference is able to break the whole argument in the proof of Theorem~\ref{thm_main_SDA_no_solves_approximate_homomorphism} if one tries to replace odd cycles with even cycles. This is a good  sanity check, as we know that $\SDP$ does solve $\CSP(\K_2)$ and, thus, $\CSP(\A)$ for any undirected bipartite graph $\A$~\cite{Barto16:sicomp} (where we denote by $\CSP(\A)$ the CSP parameterised by $\A$; i.e., $\CSP(\A)=\PCSP(\A,\A)$).
\end{rem}

Theorem~\ref{thm_description_spectral_matrix_Kneser_times_polynomial} and Proposition~\ref{prop_description_association_scheme_cycle} contain the instructions needed to design an orbital matrix corresponding to a balanced $\SDP$-matrix---which shall be concretely done in the proof of Theorem~\ref{thm_main_SDA_no_solves_approximate_homomorphism}. 
In order to show that Kneser graphs are fooling instances for $\SDA$ applied to the approximate graph homomorphism problem, this $\SDP$-matrix should be augmented with a suitable $\AIP$-matrix, as prescribed by Proposition~\ref{prop_matrix_char_SDP_new_form}.
The key to make the augmentation possible is to assign to each member of the cycle scheme an integral matrix whose support is included in the corresponding orbital, in a way that the row- and column-sum vectors are equal and constant over the whole scheme. The next result shows that such an assignment does exist.
\begin{prop}
\label{prop_fooling_matrices_cycle}
    For any odd integer $n\geq 3$ there exists a function $f:\Orb(\C_n)\to\Z^{n\times n}$ such that $\supp(f(\omega))\;\subseteq\; \omega$ and $f(\omega)\bone=f(\omega)^T\bone=\be_0$ for each $\omega\in\Orb(\C_n)$.
\end{prop}
\begin{proof}
Let $\{0,\dots,n-1\}$ be the vertex set of $\C_n$, and take $\omega\in\Orb(\C_n)$. If $\omega$ is the diagonal orbital, we let $f(\omega)=\be_0\be_0^T$, which clearly satisfies the requirements. If $\omega$ is not the diagonal orbital, the description of $\Orb(\C_n)$ given at the beginning of this section implies that there exists $j\in [m]$ such that $\omega=\{(x,y)\in \{0,\dots,n-1\}^2:\dist(x,y)=j\}$. For each vertex $x\in \{0,\dots,n-1\}$, there exist exactly two vertices $y\neq y'\in \{0,\dots,n-1\}$ such that $\dist(x,y)=\dist(x,y')=j$. In other words, $\omega$ is the edge set of an undirected graph $\HH_\omega$ all of whose vertices have degree two. It follows that $\HH_\omega$ is the disjoint union of $n_\omega$ undirected cycles, for some $n_\omega\in\N$. Let $\C$ be one of these cycles and let $x$ be a vertex belonging to $\C$.
The length of $\C$ is the minimum $\ell\in\N$ such that $x+\ell j=x$ $\mod n$. Since this quantity does not depend on $x$, we deduce that all cycles in $\HH_\omega$ have the same length $\ell_\omega\geq 3$. As $n_\omega \ell_\omega =n$, it follows in particular that $\ell_\omega$ is odd. Choose as $\C$ the $\ell_\omega$-cycle in $\HH_\omega$ containing the vertex $0$,
and relabel the vertices of $\C$ as $0=0',1',\dots,(\ell_\omega-1)'$, in the natural order.
Consider the oriented graph $\tilde\C$ obtained by setting an alternating orientation for each edge of $\C$, starting from $(0',1')$; i.e., 
\begin{align*}
\Eset(\tilde\C)=\{(0',1'),(2',1'),(2',3'),(4',3'),(4',5'),\dots,((\ell_\omega-1)',(\ell_\omega-2)'),((\ell_\omega-1)',0')\}.
\end{align*}
Let $S_+=\{(0',1'),(2',3'),\dots,((\ell_\omega-1)',0')\}$ and $S_-=\Eset(\tilde\C)\setminus S_+$. We define $f(\omega)$ as the $n\times n$  matrix whose $(x,y)$-th entry is $1$ if $(x,y)\in S_+$, $-1$ if $(x,y)\in S_-$, and $0$ otherwise. Notice that 
\begin{align*}
\supp(f(\omega))
\spaceeq
S_+\cup S_-
\spaceeq
\Eset(\tilde\C)\;\;\subseteq\;\;\Eset(\C)\;\;\subseteq\;\;\Eset(\HH_\omega)\spaceeq\omega.
\end{align*}
With the exception of $0$, each vertex of $\C_n$ either is the tail of exactly two directed edges in $\tilde\C$, whose contributions in $f(\omega)$ have opposite signs, or it is not the tail of any directed edge in $\tilde\C$. As for $0=0'$, it is the tail of exactly one directed edge, whose contribution is $+1$. The same statement is true if we replace ``tail'' with ``head''. As a consequence, we find that $f(\omega)\bone=f(\omega)^T\bone=\be_0$, thus concluding the proof.
\end{proof}

\section{A lower bound against SDA}
\label{subsec_lower_bound_on_SDA}
We now have all the ingredients for proving the main result of the paper.

\begin{thm*}
[Theorem~\ref{thm_main_SDA_no_solves_approximate_homomorphism} restated]
Let $\A,\B$ be non-bipartite loopless undirected graphs such that $\A\to\B$. Then
$\SDA$ does not solve $\PCSP(\A,\B)$.
\end{thm*}
\begin{proof}
Observe that, for $\A$ and $\B$ as in the statement of the theorem,
there exist $n,n'\geq 3$ with $n$ odd such that $\C_n\to\A$ and $\B\to\K_{n'}$. (For example, we may choose $n$ and $n'$ as the odd girth of $\A$ and the chromatic number of $\B$, respectively.)
Let $m=\frac{n-1}{2}$, 
and let $\tilde P$ be the character table of the association scheme corresponding to $\Orb(\C_n)$. 
By Proposition~\ref{prop_description_association_scheme_cycle}, there exists $0<\delta<2$ such that, up to a permutation of the rows, $\tilde P\be_0=\bone$ and
$\tilde P\be_1=
    \begin{bmatrix}
        2\\
        \bz
    \end{bmatrix}$
for some vector $\bz\in\R^m$ all of whose entries have absolute value strictly smaller than $\delta$. Without loss of generality, we can assume that $\delta$ is rational. Let $t\in\N$ be such that $t\geq \frac{2n'}{2-\delta}$ and $\frac{t}{\delta}\in\N$, and let $s=\frac{2t}{\delta}+t$. Observe that $s>2t$. We claim that $\SDA(\GG_{s,t},\C_n)=\YES$. Since, as shown in Proposition~\ref{prop_monotonicity_SDP_SDA_wrt_homomorphisms}, $\SDA$ is monotone with respect to the homomorphism preorder of the arguments, this would imply that $\SDA(\GG_{s,t},\A)=\YES$. However, using Lov\'{a}sz's formula for the chromatic number of Kneser graphs~\cite{lovasz1978kneser}, we find
\begin{align*}
    \chi(\GG_{s,t})
    \spaceeq
    s-2t+2
    \spaceeq
    \frac{2t}{\delta}+t-2t+2
    \spaceeq
    \frac{t(2-\delta)}{\delta}+2
    \;\;\geq\;\;
    \frac{2n'}{\delta}+2
    \;\;>\;\;
    n'+2.
\end{align*}
This means that $\GG_{s,t}\not\to\K_{n'}$ and, hence, $\GG_{s,t}\not\to\B$. As a consequence, the truth of the claim would establish that $\SDA$ does not solve $\PCSP(\A,\B)$, thus concluding the proof of the theorem.

Let $P$ be the character table of $\GG_{s,t}$, and recall that $\bh$ denotes the vector $(0,1,\dots,t)$ (which, as usual, we view as a column vector).
Consider the matrices
\begin{align*}
    W=\begin{bmatrix}
        \bone-\frac{1}{t}\bh & \frac{1}{t}\bh & O_{t+1,m-1}
    \end{bmatrix}\in\R^{(t+1)\times (m+1)},\quad\quad
    K=\frac{1}{2n}\diag(2,1,\dots,1)\in\R^{(m+1)\times (m+1)},
\end{align*}
and $V=WK$. We now show that $V$ meets the conditions in Corollary~\ref{cor_critierion_acceptance_SDP} and, thus, it is the orbital matrix for a balanced $\SDP$-matrix.
Recall that the diagonal orbitals of $\GG_{s,t}$ and $\C_n$ are those having index $0$, while the (unique) edge orbitals of $\GG_{s,t}$ and $\C_n$ are those having index $t$ and $1$, respectively. Since $v_{i,j}=0$ whenever $i=0, j\neq 0$ or $i=t,j\neq 1$, the conditions (\ref{c3}) and (\ref{c4}) are satisfied. Observe that $\bmu^{\C_n}$ is the vector $2n\bone-n\be_0$.
Therefore, $V\bmu^{\C_n}=WK\bmu^{\C_n}=W\bone=\bone$, so (\ref{c2}) holds, too.
Theorem~\ref{thm_description_spectral_matrix_Kneser_times_polynomial} yields
{\allowdisplaybreaks
\begin{align*}
PW
&\spaceeq
\begin{bmatrix}
    P\bone-\frac{1}{t}P\bh & \frac{1}{t}P\bh & O_{t+1,m-1}
\end{bmatrix}\\
&\spaceeq
\bin{s}{t}
\begin{bmatrix}
\be_0-(1-\frac{t}{s})\be_0+\frac{s-t}{s^2-s}\be_1 &
(1-\frac{t}{s})\be_0-\frac{s-t}{s^2-s}\be_1 & O_{t+1,m-1}
\end{bmatrix}\\
&\spaceeq
\bin{s}{t}
\begin{bmatrix}
\frac{t}{s} & 1-\frac{t}{s} & 0 & \dots & 0 \\
\frac{s-t}{s^2-s} & \frac{t-s}{s^2-s} & 0 & \dots & 0\\
0 & 0 & 0 & \dots & 0\\
\vdots & \vdots & \vdots & \vdots & \vdots\\
0 & 0 & 0 & \dots & 0
\end{bmatrix}.
\end{align*}
}
It follows that
\begin{align*}
    PV\tilde P^T
    &\spaceeq
    PWK\tilde P^T
    \spaceeq
    \frac{\bin{s}{t}}{2n}
    \begin{bmatrix}
\frac{2t}{s} & 1-\frac{t}{s}\\
\frac{2(s-t)}{s^2-s} & \frac{t-s}{s^2-s}\\
0 & 0\\
\vdots & \vdots\\
0 & 0
\end{bmatrix}
\begin{bmatrix}
    1 & \bone_m^T\\
    2 & \bz^T
\end{bmatrix}
\spaceeq
\frac{\bin{s}{t}}{2n}
\begin{bmatrix}
    2 & \frac{2t}{s}\bone_m^T+(1-\frac{t}{s})\bz^T\\
    0 & \frac{s-t}{s^2-s}(2\cdot\bone_m-\bz)^T
\end{bmatrix}.
\end{align*}
We have $2\cdot\bone_m-\bz>0$. Using that $\frac{t}{s}=\frac{\delta}{2+\delta}$ and $z_i>-\delta$ for each $i\in [m]$, we find that
\begin{align*}
    \frac{2t}{s}+\left(1-\frac{t}{s}\right)z_i
    \spaceeq
    \frac{2\delta}{2+\delta}+\left(1-\frac{\delta}{2+\delta}\right)z_i
    \;\;>\;\;
    \frac{2\delta}{2+\delta}+\left(1-\frac{\delta}{2+\delta}\right)(-\delta)
    \spaceeq
    0,
\end{align*}   
thus showing that $\frac{2t}{s}\bone_m+(1-\frac{t}{s})\bz>0$. It follows that $PV\tilde P^T\geq 0$, which means that (\ref{c1}) is met. Applying Corollary~\ref{cor_critierion_acceptance_SDP}, we deduce that $\SDP(\GG_{s,t},\C_n)=\YES$ and that the matrix 
\begin{align*}
M=\sum_{\substack{\omega\in\Orb(\GG_{s,t})\\\tilde\omega\in\Orb(\C_n)}} v_{\omega\tilde\omega}\,R_\omega\otimes R_{\tilde\omega}    
\end{align*}
is a balanced $\SDP$-matrix for $\GG_{s,t}$, $\C_n$, cf.~Proposition~\ref{prop_orbital_matrix_of_balanced_matrix}.

The next step is to add $\AIP$. For each $\bx\in\Vset(\GG_{s,t})^2$, let $\omega^{(\bx)}$ be the orbital of $\GG_{s,t}$ containing $\bx$, and choose an orbital $\tilde\omega^{(\bx)}$ of $\C_n$ satisfying $v_{\omega^{(\bx)}\tilde\omega^{(\bx)}}\neq 0$.
Letting $f:\Orb(\C_n)\to\Z^{n\times n}$ be the function from Proposition~\ref{prop_fooling_matrices_cycle}, we consider the $\bin{s}{t}n\times \bin{s}{t}n$ matrix $N$ defined by $N_\bx=f(\tilde\omega^{(\bx)})$ for each $\bx$ (where $N_\bx=(\be_{x_1}\otimes I_n)^TN(\be_{x_2}\otimes I_n)$ is the $\bx$-th block of $N$).
We claim that $N$ is an $\AIP$-matrix for $\GG_{s,t},\C_n$. Note that, if $\bx=(x,x)\in\Vset(\GG_{s,t})^2$, we have $\omega^{(\bx)}=\omega_0$ and, thus, $\tilde\omega^{(\bx)}=\tilde\omega_0$, which gives $\supp(N_\bx)=\supp(f(\tilde\omega_0))\subseteq\tilde\omega_0$. Similarly, if $\bx\in\Eset(\GG_{s,t})$, then $\omega^{(\bx)}=\omega_t$ and, thus, $\tilde\omega^{(\bx)}=\tilde\omega_1$, which gives $\supp(N_\bx)=\supp(f(\tilde\omega_1))\subseteq\tilde\omega_1=\Eset(\C_n)$. This yields the conditions (\ref{r1}) and (\ref{r2}). Moreover, for $\bx=(x_1,x_2)\in\Vset(\GG_{s,t})^2$,
we find
\begin{align*}
    (\be_{x_1}\otimes I_n)^TN(\be_{x_2}\otimes\bone_n)
    &\spaceeq
     (\be_{x_1}\otimes I_n)^TN(\be_{x_2}\otimes I_n)(1\otimes\bone_n)
     \spaceeq
     N_{\bx}\bone_n
     \spaceeq
     f(\tilde\omega^{(\bx)})\bone_n
\end{align*}
which, by the properties of $f$, is constant over the orbitals of $\C_n$; this gives (\ref{r3}). Similarly, using that $f(\tilde\omega^{(\bx)})^T\bone_n$ is constant over the orbitals, we obtain (\ref{r4}). Finally, (\ref{r5}) follows by observing that $\bone_n^T N_\bx\bone_n=\bone_n^T f(\tilde\omega^{(\bx)})\bone_n=\bone_n^T\be_0=1$ for any $\bx$. As a consequence, $N$ is a relaxation matrix; since its entries are integral, it is an $\AIP$-matrix. 
For any $\bx\in\Vset(\GG_{s,t})^2$, the $\bx$-th block of $M$ satisfies
\begin{align*}
    M_\bx
    &\spaceeq
    (\be_{x_1}\otimes I_n)^TM(\be_{x_2}\otimes I_n)
    \spaceeq
    \sum_{\substack{\omega\in\Orb(\GG_{s,t})\\\tilde\omega\in\Orb(\C_n)}}v_{\omega\tilde\omega}\,(\be_{x_1}\otimes I_n)^T(R_\omega\otimes R_{\tilde\omega})(\be_{x_2}\otimes I_n)\\
    &\spaceeq
    \sum_{\substack{\omega\in\Orb(\GG_{s,t})\\\tilde\omega\in\Orb(\C_n)}}v_{\omega\tilde\omega}\,(\be_{x_1}^TR_\omega\be_{x_2})R_{\tilde\omega}
    \spaceeq
    \sum_{\tilde\omega\in\Orb(\C_n)}v_{\omega^{(\bx)}\tilde\omega}R_{\tilde\omega}.
\end{align*}
Since $v_{\omega^{(\bx)}\tilde\omega^{(\bx)}}\neq 0$, using that the orbitals of a graph are disjoint, we deduce that $R_{\tilde\omega^{(\bx)}}\;\triangleleft\; M_\bx$. On the other hand, we have $\supp(N_\bx)=\supp(f(\tilde\omega^{(\bx)}))\subseteq\tilde\omega^{(\bx)}=\supp(R_{\tilde\omega^{(\bx)}})$, which means that $N_\bx\;\triangleleft\;R_{\tilde\omega^{(\bx)}}$. It follows that $N\circ((I_{\bin{s}{t}}+\adj(\GG_{s,t}))\otimes J_n)\;\triangleleft\; N\;\triangleleft\; M$. Applying Proposition~\ref{prop_matrix_char_SDP_new_form}, we conclude that $\SDA(\GG_{s,t},\C_n)=\YES$, as required.
\end{proof}

We note that the SDP part of the integrality gap in Theorem~\ref{thm_main_SDA_no_solves_approximate_homomorphism} may be directly converted into Unique-Games approximation hardness of AGH through Raghavendra's framework~\cite{Raghavendra08:everycsp}. 
Given two digraphs $\X,\X'$ and a real number $0\leq\epsilon\leq 1$, an \emph{$\epsilon$-homomorphism} from $\X$ to $\X'$ is a map $f:\Vset(\X)\to\Vset(\X')$ that preserves at least $(1-\epsilon)$-fraction of the edges of $\X$. A \emph{robust} algorithm for $\PCSP(\A,\B)$ is an algorithm that finds a $g(\epsilon)$-homomorphism from $\X$ to $\B$ whenever the instance $\X$ is such that there exists an $\epsilon$-homomorphism from $\X$ to $\A$, where $g$ is some monotone, nonnegative function satisfying $g(\epsilon)\to 0$ as $\epsilon\to 0$.
As observed in~\cite{bgs_robust23stoc},
it follows from~\cite{Raghavendra08:everycsp} that any PCSP admitting a polynomial-time robust algorithm is solved by SDP, assuming the Unique Games Conjecture (UGC) of~\cite{Khot02stoc}. Thus, Theorem~\ref{thm_main_SDA_no_solves_approximate_homomorphism} implies the following conditional hardness result for AGH. 
\begin{cor}
    Let $\A,\B$ be non-bipartite loopless undirected graphs such that $\A\to\B$. Then, assuming the $\operatorname{UGC}$ and $\operatorname{P}\neq\operatorname{NP}$, $\PCSP(\A,\B)$ does not admit a polynomial-time robust algorithm.
\end{cor}

\section{Incomparability to the BA hierarchy}
\label{section_incomparability}
The $\BLP+\AIP$ algorithm---whose name we abbreviate to $\BA$ in this paper---was introduced in~\cite{bgwz20} as a combination of two standard algorithmic techniques for CSPs: the \emph{basic linear programming relaxation} $\BLP$ and the affine integer programming relaxation $\AIP$ (the same we use in the current work as the linear Diophantine part of $\SDA$). Unlike $\SDA$, this relaxation does not solve all bounded-width CSPs, as noted in~\cite{bgwz20}.
Consequently, $\BA$ is strictly less powerful than $\SDA$. 
It is possible to progressively strengthen the $\BA$ algorithm through the
\emph{lift-and-project} technique~\cite{Laurent03}, which results in the so-called \emph{$\BA$ hierarchy}~\cite{cz23soda:minions}. The \emph{$k$-th level} of the hierarchy, denoted by $\BA^k$, corresponds to applying $\BA$ to a modified instance whose variables are sets of variables of the original instance of size up to $k$.\footnote{For the explicit definitions of $\BLP$, $\BA$, and $\BA^k$, we refer the reader to~\cite{BBKO21},~\cite{bgwz20}, and~\cite{cz23soda:minions}, respectively.}
Recently,~\cite{cz23stoc:ba} established 
a lower bound against this model, by showing that no constant level of the $\BA$ hierarchy solves the approximate graph colouring problem.
Since constant levels of the $\BA$ hierarchy do solve all bounded-width CSPs, it is natural to investigate how the hierarchy compares to $\SDA$. In particular, if some level of the hierarchy dominated $\SDA$ (in a sense that we will now make formal), the lower bound in~\cite{cz23stoc:ba} would immediately imply the same lower bound against $\SDA$ (though only for the approximate graph colouring problem).
In this section, we establish that this is not the case: $\SDA$ (and, in fact, already $\SDP$ and a weaker version of the latter) is not dominated by any level of the $\BA$ hierarchy.
As a consequence, the specialisation of Theorem~\ref{thm_main_SDA_no_solves_approximate_homomorphism} to approximate graph colouring does not follow as a corollary of~\cite{cz23stoc:ba}. 

\begin{defn}
\label{defn_tests_and_annexes}
Let $\mathfrak{D}$ be the family of all (finite) digraphs. We define a \emph{test} to be a function $T:\mathfrak{D}^2\to\{\YES,\NO\}$. We say that a test $T$ is a \emph{polynomial-time} test if, for each $\A\in\mathfrak{D}$, there exists an algorithm $\Alg_\A$ taking digraphs as inputs and returning values in $\{\YES,\NO\}$, such that 
\begin{itemize}
    \item[$(i)$] $T(\X,\A)=\Alg_\A(\X)$ for each $\X\in\mathfrak{D}$, and
    \item[$(ii)$] $\Alg_\A$ can be implemented in polynomial time in the size of the input. 
\end{itemize}
Also, we say that a test $T$ is \emph{complete} if $T(\X,\A)=\YES$ for any $\X,\A\in\mathfrak{D}$ such that $\X\to\A$; i.e., a complete test has no false negatives.
We define a partial order ``$\preceq$'' on the set of tests:
Given two tests $T_1,T_2$, we write $T_1\preceq T_2$ (and we say that $T_2$ \emph{dominates} $T_1$) if, 
for any $\X,\A\in\mathfrak{D}$, $T_2(\X,\A)=\YES$ implies $T_1(\X,\A)=\YES$.
\end{defn}

All relaxations mentioned in this work are complete tests. For such tests, the fact that $T_1\preceq T_2$ means that $T_2$ is at least as powerful as $T_1$, in that it has fewer false positives.
For example, it directly follows from the definitions that $\SDA$ dominates both $\SDP$ and $\AIP$. Moreover, since any solution to $\SDP$ can be turned into a solution to $\BLP$ (by taking the norms of the vector variables $\blambda_{x,a}$), it follows that $\BLP\preceq\SDP$ and 
$\BA=\BA^1\preceq\SDA$.
If, for some $k\in\N$, we had $\SDA\preceq\BA^k$, the results in~\cite{cz23stoc:ba} would directly imply that approximate graph colouring is not solved by $\SDA$ and, moreover, that the same fooling instances produced in~\cite{cz23stoc:ba} could be used for fooling $\SDA$. In Subsection~\ref{subsec_SDPeps_vs_BAk}, we show that this is not the case, as not even a weaker, polynomial-time version of $\SDP$---the test $\SDP^\epsilon$ described in Subsection~\ref{subsec_SDPeps} below---is dominated by $\BA^k$.

\subsection{A tale of two polytopes}
\label{subsec_SDPeps}
A few years ago, O'Donnell
noted that 
polynomial-time solvability of certain semidefinite programming relaxations, assumed in several papers in the context
of the Sum-of-Squares proof system, is not known in general~\cite{ODonnell17:itcs}.
In fact,
it is a well-known open question in optimisation theory whether all semidefinite programs can be solved to near-optimality in polynomial time~\cite{Ramana97:mp,Gartner2012approximation}. 
To the best of the authors' knowledge, details of how the semidefinite
program SDP can be efficiently solved to
near-optimality (if at all) have not been made explicit in the literature. This motivates us to
give a formal argument showing that this is indeed possible. As we shall
see, the issue requires some unexpected matrix-theoretic considerations.

It is well known that a polynomial-time algorithm (in the Turing model of computation) for semidefinite programming based on the ellipsoid method exists \cite{grotschel1981ellipsoid,MR1261419,Vandenberghe96} under the assumption that the feasible region contains a ``large enough'' inner ball and is contained in a ``small enough'' outer ball---a requirement known as \emph{Slater condition}.\footnote{Another polynomial-time algorithm is based on the interior-point methods~\cite{de2016turing}.}
In this subsection, we show that the semidefinite program SDP can be solved to near-optimality in polynomial time, by reformulating it as an optimisation problem meeting Slater condition.

Recall that $\Frob{M}{N}=\Tr(M^TN)$ denotes the Frobenius inner product of matrices, and let $\|M\|_{\operatorname{F}}=\sqrt{\Frob{M}{M}}$ denote the corresponding norm. 
Given a set $\mathscr{M}$ of square matrices of equal size, a matrix $M\in \mathscr{M}$, and a real number ${r}>0$, we consider the ball $\mathscr{B}_\mathscr{M}(M;{r})=\{N\in \mathscr{M}:\|N-M\|_{\operatorname{F}}<{r}\}$. Throughout this and the next subsections, we shall denote the cone of positive semidefinite matrices by $\mathscr P$.
For $\ell,m\in\N$,
let $C,A_1,\dots,A_m$ be rational $\ell\times\ell$ matrices, and let $b_1,\dots,b_m$ be rational numbers. We denote by $\mathscr{V}$ the polytope containing all real symmetric $\ell\times\ell$ matrices $M$ satisfying $\Frob{A_i}{{M}}\leq b_i$ for each $i\in [m]$.
Consider the semidefinite program in standard form
\begin{align}
\label{eqn_1246_27_06}
\begin{array}{ll}
     \inf&\Frob{C}{{M}}  \\
     \mbox{subject to}&{M}\in\mathscr{V}\cap\mathscr{P}
\end{array}
\end{align}
and let $\nu$ be the optimal value of the program.
Let also $\mathscr{V}^{\operatorname{a}}$ denote the affine hull of $\mathscr{V}$, i.e., 
the intersection of all affine spaces containing $\mathscr{V}$ (where a set $\mathscr{S}$ of $\ell\times\ell$ real matrices is an affine space if $\lambda M+(1-\lambda)N\in \mathscr{S}$ whenever $M,N\in \mathscr{S}$ and $\lambda\in\R$).
For rationals $r,R>0$, we say that a matrix ${M_0}\in\mathscr{V}\cap\mathscr{P}$ is an \emph{$(r,R)$-Slater point} for~\eqref{eqn_1246_27_06} if 
$\mathscr{B}_{\mathscr{V}^{\operatorname{a}}}({M_0};r)\;\subseteq\;\mathscr{V}\cap\mathscr{P}\;\subseteq\;\mathscr{B}_{\mathscr{V}^{\operatorname{a}}}({M_0};R)$. 
The next result from~\cite{grotschel1981ellipsoid} (see also the formulation in~\cite{de2016turing}) establishes that the ellipsoid method 
can be used to solve a semidefinite program up to arbitrary precision in polynomial time \emph{provided that} there exists a Slater point. 
\begin{thm}[\cite{grotschel1981ellipsoid}]
\label{thm_ellipsoid_method_is_fast_with_Slater}
    Let ${M_0}$ be an $(r,R)$-Slater point for the semidefinite program~\eqref{eqn_1246_27_06}. Then for any rational $\epsilon>0$ one can find a rational matrix ${M}^*\in\mathscr{V}\cap\mathscr{P}$ such that $\Frob{C}{{M}^*}-\nu\leq\epsilon$ in time polynomial in $\ell$, $m$, $\log(R/r)$, $\log(1/\epsilon)$, and the bit-complexity of the input data $C$, $A_i$, $b_i$, and $M_0$.
\end{thm}
Our goal is then to reformulate the system~\eqref{eqns_SDP} as a program in the form~\eqref{eqn_1246_27_06}, and to find for it a suitable Slater point. First of all, observe that we cannot simply introduce a dummy objective function to be minimised over the feasible set of ~\eqref{eqns_SDP} (i.e., the set of solutions to (\ref{SDP1})--(\ref{SDP4})), as this set can be empty, in which case it clearly contains no Slater points.\footnote{In fact, the set is nonempty precisely when $\SDP(\X,\A)=\YES$.} The natural choice is then to relax the condition (\ref{SDP3})---which requires that the solution should be compatible with the edge sets of $\X$ and $\A$---by turning it into an objective function to be minimised: Given a $pn\times pn$ matrix $M$, we let 
\begin{align}
\label{eqn_f_of_M}
f(M)=
\sum_{(x,y)\in\Eset(\X)}\sum_{(a,b)\in \Vset(\A)^2\setminus\Eset(\A)}(\be_x\otimes\be_a)^TM(\be_y\otimes\be_b).
\end{align}
(Notice that we are working with the matrix formulation of~\eqref{eqns_SDP}, which is compatible with the standard form~\eqref{eqn_1246_27_06}.) This is sufficient to make the feasible set nonempty, as is witnessed, for example, by the positive semidefinite matrix $\frac{1}{n}J_p\otimes I_n$. We now need to declare which polytope takes the role of $\mathscr{V}$ in the standard form~\eqref{eqn_1246_27_06}. 
This is an important choice: The definition of Slater points takes into account not only the feasible set $\mathscr{V}\cap\mathscr{P}$ of a program, but also the polytope $\mathscr{V}$ involved in its formulation. Hence, it might happen that Slater condition can be enforced by modifying the formulation of a program in a way that the dimension of the polytope $\mathscr{V}$ is reduced, while still preserving the feasible set $\mathscr{V}\cap\mathscr{P}$.
In the current setting,
one natural candidate is the polytope described by taking the constraints of ~\eqref{eqns_SDP} and discarding (\ref{SDP3}) and positive semidefiniteness---i.e., in the matrix formulation, the polytope $\mathscr{W}$ containing all $pn\times pn$ symmetric entrywise-nonnegative matrices satisfying the conditions (\ref{r1}) (``diagonal blocks are diagonal'') and (\ref{r6}) (``the entries in each block sum up to $1$''). Another natural choice consists in looking at the conditions defining an $\SDP$-matrix, and discarding the condition (\ref{r2}) and positive semidefiniteness: We let $\mathscr{U}$ denote the polytope of $pn\times pn$ symmetric entrywise-nonnegative matrices satisfying the conditions (\ref{r1}), (\ref{r3}), (\ref{r4}), and (\ref{r5}).
The two choices result in two different programs:
\vspace{-.4cm}
\begin{center}
    \begin{minipage}{.4\textwidth}
        \begin{align}
        \label{eqn_SDP_W}
        \tag{SDP$'$}
        \begin{array}{ll}
             \inf&f(M)  \\
             \mbox{subject to}&M\in \mathscr{W}\cap\mathscr{P}
        \end{array}
        \end{align}
    \end{minipage}
\hspace{.15\textwidth}
    \begin{minipage}{.4\textwidth}
        \begin{align}
        \label{eqn_SDP_U}
        \tag{SDP$''$}
        \begin{array}{ll}
             \inf&f(M)  \\
             \mbox{subject to}&M\in \mathscr{U}\cap\mathscr{P}.
        \end{array}
        \end{align}
    \end{minipage}
\end{center}
It follows from Proposition~\ref{prop_equivalence_some_conditions_SDP} that $\mathscr{U}\subseteq \mathscr{W}$ and $\mathscr{U}\cap\mathscr P= \mathscr{W}\cap\mathscr P$. In particular, this means that~\eqref{eqn_SDP_W} and~\eqref{eqn_SDP_U} are two different formulations of \emph{the same minimisation problem}. Nevertheless,
the two propositions below show that \emph{only the second} formulation of the program meets Slater condition, which guarantees the existence of a polynomial-time algorithm solving it to near-optimality in the Turing model of computation.

\begin{prop}
\label{prop_Slater_condition_for_SDP}
    There exists ${M_0}\in \mathscr{U}\cap\mathscr P$ such that $\mathscr{B}_{\mathscr{U}^{\operatorname{a}}}({M_0};{\frac{1}{n^2}})\,\subseteq\,\mathscr{U}\cap\mathscr{P}\,\subseteq\,\mathscr{B}_{\mathscr{U}^{\operatorname{a}}}({M_0};2p^2+1)$.
\end{prop}
\begin{prop}
\label{prop_no_Slater_condition_for_SDP_bis}
    If $p,n\geq 2$, $\mathscr{B}_{\mathscr{W}^{\operatorname{a}}}(M_0;{r})\,\not\subseteq\,\mathscr{P}$ for any $M_0\in \mathscr{W}\cap\mathscr P$ and ${r}>0$.
\end{prop}
\noindent For this reason, we define the test $\SDP^{\epsilon}$ using the formulation~\eqref{eqn_SDP_U}. More precisely, for $\epsilon>0$, $\SDP^\epsilon$ is described as follows:
\begin{itemize}
    \item Take two digraphs $\X,\A$ as input;
    \item run the ellipsoid method~\cite{grotschel1981ellipsoid} 
    on the program~\eqref{eqn_SDP_U} with precision $\epsilon$, obtaining an output $M^\ast\in\mathscr{U}\cap\mathscr{P}$;
    \item if $f(M^\ast)\leq\epsilon$, set $\SDP^{\epsilon}(\X,\A)=\YES$; otherwise, set $\SDP^{\epsilon}(\X,\A)=\NO$.
\end{itemize}
\noindent We thus obtain the following result.
\begin{thm}
\label{thm_SDP_eps_complete_pol}
For each $\epsilon>0$, $\SDP^{\epsilon}$ is a complete, polynomial-time test. Moreover, $\SDP^{\epsilon}\preceq\SDP$.
\end{thm}
\begin{proof}
    It follows from Proposition~\ref{prop_Slater_condition_for_SDP} that the program~\eqref{eqn_SDP_U} has a $(\frac{1}{n^2},2p^2+1)$-Slater point.
    Using Theorem~\ref{thm_ellipsoid_method_is_fast_with_Slater}, we deduce that we can find a near-optimal solution to~\eqref{eqn_SDP_U} up to any given precision $\epsilon>0$ in time polynomial in the sizes of $\X$ and $\A$. In particular, if we fix $\A$, $\SDP^\epsilon$ can be implemented in polynomial time in the size of $\X$ and it is thus a polynomial-time test, as per Definition~\ref{defn_tests_and_annexes}. 
    Moreover, if $\SDP(\X,\A)=\YES$, the optimal value of~\eqref{eqn_SDP_U} is $0$. As a consequence, the solution $M^\ast$ found by the ellipsoid method satisfies $f(M^\ast)\leq\epsilon$ (cf.~Theorem~\ref{thm_ellipsoid_method_is_fast_with_Slater}), which means that $\SDP^\epsilon(\X,\A)=\YES$. It follows that $\SDP^\epsilon\preceq\SDP$. In particular,
    since $\SDP$ is complete, this implies that $\SDP^\epsilon$ is also complete.
\end{proof}

In Subsection~\ref{subsec_SDPeps_vs_BAk}, this weaker, polynomial-time version of $\SDP$ will prove to be strong enough to correctly classify cliques and, therefore, not to be dominated by the $\BA$ hierarchy. We now give a proof of Proposition~\ref{prop_Slater_condition_for_SDP}, that finds a Slater point for the program~\eqref{eqn_SDP_U}. The following, simple description of the association schemes corresponding to cliques shall be useful. 

\begin{rem}
\label{rem_association_scheme_cliques}
It is straightforward to check that, for any $n\geq 2$, the clique $\K_n$ is generously transitive, and the association scheme corresponding to $\Orb(\K_n)$ consists of the two matrices $I_n$ and $J_n-I_n$. Either by a direct computation or noting that $\K_n=\GG_{n,1}$, we find that the character table of $\K_n$ is the matrix
$\begin{bmatrix}
    1&n-1\\1&-1
\end{bmatrix}$.
\end{rem}

\begin{proof}[Proof of Proposition~\ref{prop_Slater_condition_for_SDP}]
    Let ${M_0}=\frac{1}{n}I_p\otimes I_n+\frac{1}{n^2}(J_p-I_p)\otimes J_n$, and notice that ${M_0}\in \mathscr{U}$.
    Letting $\omega_0$, $\omega_1$ denote the diagonal and edge orbitals of $\K_p$ and $\tilde\omega_0$, $\tilde\omega_1$ denote the diagonal and edge orbitals of $\K_n$, we see that ${M_0}$ may be written in the form~\eqref{eqn_1334_0504} with $v_{\omega_0\tilde\omega_0}=\frac{1}{n}$, $v_{\omega_0\tilde\omega_1}=0$, and $v_{\omega_1\tilde\omega_0}=v_{\omega_1\tilde\omega_1}=\frac{1}{n^2}$. It follows from Proposition~\ref{prop_orbital_matrix_of_balanced_matrix} that ${M_0}$ is balanced for $\K_p$, $\K_n$, with orbital matrix $V=\begin{bmatrix}
        \frac{1}{n}&0\\\frac{1}{n^2}&\frac{1}{n^2}
    \end{bmatrix}$. Let $P$ and $\tilde P$ be the character tables of $\K_p$ and $\K_n$, respectively, and recall their expressions from Remark~\ref{rem_association_scheme_cliques}. Using Theorem~\ref{prop_decomposing_spectrum_M_orbital}, we deduce that the spectrum of ${M_0}$ consists of the entries of the matrix
    \begin{align}
    \label{eqn_2606_1633}
    PV\tilde P^T
    \spaceeq
    \begin{bmatrix}
        1&p-1\\1&-1
    \end{bmatrix}
    \begin{bmatrix}
        \frac{1}{n}&0\\\frac{1}{n^2}&\frac{1}{n^2}
    \end{bmatrix}
    \begin{bmatrix}
        1&1\\n-1&-1
    \end{bmatrix}
    \spaceeq
    \begin{bmatrix}
        \frac{p}{n}&\frac{1}{n}\\0&\frac{1}{n}
    \end{bmatrix}.
    \end{align}
    In particular, ${M_0}\in\mathscr{P}$. 
    
    In order to show that $\mathscr{U}\cap\mathscr{P}\subseteq\mathscr{B}_{\mathscr{U}^{\operatorname{a}}}({M_0};2p^2+1)$,
    notice that any matrix $M\in\mathscr{U}$ satisfies
    $
        \|M\|_{\operatorname{F}}^2
        =
        \Tr(M^2)
        \leq
        (\bone^TM\bone)^2
        =
        p^4$.
    Therefore, 
    \begin{align*}
        \|M-M_0\|_{\operatorname{F}}
        \;\;\leq\;\;
        \|M\|_{\operatorname{F}}+\|M_0\|_{\operatorname{F}}
        \;\;\leq\;\;
        2p^2
        \;\;<\;\;
        2p^2+1,
    \end{align*}
    as needed.
    
    We now need to prove that $\mathscr{B}_{\mathscr{U}^{\operatorname{a}}}({M_0};{\frac{1}{n^2}})\subseteq\mathscr{U}\cap\mathscr{P}$.
    Let $Z=\spann(\{(\be_x-\be_y)\otimes\bone_n:x,y\in \Vset(\K_p)\})$, and consider two matrices $Q_1$ and $Q_2$ whose columns form orthonormal bases for $Z$ and $Z^\perp$, respectively, where $Z^\perp$ denotes the orthogonal complement of $Z$ in $\R^{pn}$. 
    Let $\mathscr{H}$ be the vector space of  $pn\times pn$ symmetric matrices satisfying (\ref{r3}). Since $\mathscr{H}$ is in particular an affine space and $\mathscr{U}\subseteq\mathscr{H}$, we have $\mathscr{U}^{\operatorname{a}}\subseteq\mathscr{H}$.
    Given any $A\in\mathscr{H}$, let $\widehat{A}=Q_2^TAQ_2$.
    The condition (\ref{r3})
     implies that $Z\subseteq \ker(A)$, which means that $AQ_1=O$. Letting $Q=\begin{bmatrix}
        Q_1&Q_2
    \end{bmatrix}\in\R^{pn\times pn}$, we deduce that 
    \begin{align}
    \label{eqn_2023_07_06}
    Q^TAQ
    \spaceeq
    \begin{bmatrix}
        Q_1^TAQ_1&Q_1^TAQ_2\\Q_2^TAQ_1&Q_2^TAQ_2
    \end{bmatrix}
    \spaceeq
    \begin{bmatrix}
        O&O\\O&Q_2^TAQ_2
    \end{bmatrix}
    \spaceeq
    \begin{bmatrix}
        O&O\\O&\widehat{A}
    \end{bmatrix}.
    \end{align}
    We claim that $Z=\ker({M_0})$. The inclusion $Z\subseteq\ker(M_0)$ is clear from the fact that $M_0\in\mathscr{H}$.
Take $\bw\in Z^\perp$, and notice that this implies that there exists a constant $c$ for which $(\be_x\otimes \bone_n)^T\bw=c$ for every $x\in \Vset(\K_p)$. If $\bw\in\ker({M_0})$, we have
\begin{align*}
    \bzero
    \spaceeq
    {M_0}\bw
    \spaceeq
    \frac{1}{n}\bw+\frac{c(p-1)}{n^2}\bone_{pn}.
\end{align*}
It follows that $\bw=\frac{c(1-p)}{n}\bone_{pn}$, which gives, for any $x\in \Vset(\K_p)$, 
\begin{align*}
    c
    \spaceeq
    (\be_x\otimes \bone_n)^T\bw
    \spaceeq
    \frac{c(1-p)}{n}(\be_x\otimes \bone_n)^T\bone_{pn}
    \spaceeq
    c-pc,
\end{align*}
whence we find $c=0$ and, thus, $\bw=\bzero$. It follows that $\ker({M_0})\cap Z^\perp=\{\bzero\}$, which yields the claimed identity $Z=\ker({M_0})$.

Take  $N\in\mathscr{B}_{\mathscr{U}^{\operatorname{a}}}({M_0};\frac{1}{n^2})$; we need to show that $N\in\mathscr{U}\cap\mathscr{P}$. Observe that $N\in\mathscr{H}$, so $\widehat{N}$ is well defined. We claim that $\widehat{N}$ is a positive definite matrix. By~\eqref{eqn_2023_07_06}, this would imply that $Q^TNQ\in\mathscr{P}$ and, thus, that $N\in\mathscr{P}$.
Since $\mathscr{H}$ is a vector space, we have $N-M_0\in\mathscr{H}$; moreover,  $\widehat{N-{M_0}}=\widehat{N}-\widehat{{M_0}}$.
Let $q=\dim(Z)=\dim(\ker({M_0}))$, and take a vector $\bv\in\R^{pn-q}$ having unitary norm. Order the eigenvalues of ${M_0}$ as $\lambda_1({M_0})\leq\dots\leq\lambda_{pn}({M_0})$.
From~\eqref{eqn_2606_1633}, we see that $0=\lambda_{q}({M_0})<\lambda_{q+1}({M_0})=\frac{1}{n}$. Using the Courant--Fischer variational characterisation of the spectrum of symmetric matrices (see~\cite[\S~8.2]{hogben2013handbook}),
we deduce that
\begin{align*}
\bv^T\widehat{{M_0}}\bv
\spaceeq
(Q_2\bv)^T{M_0}(Q_2\bv)
\geq
\min_{\substack{\by\in Z^\perp\\\by^T\by=1}}\by^T{M_0}\by
\spaceeq
\lambda_{q+1}({M_0})
\spaceeq \frac{1}{n}.
\end{align*}  
Moreover, letting $\|\cdot\|_2$ denote the spectral matrix norm, we have
    \begin{align*}
        \lvert\bv^T\widehat{N-{M_0}}\bv\rvert
        \;\;&\leq\;\;
        \|\widehat{N-{M_0}}\|_2
        \;\;\leq\;\;
        \|N-{M_0}\|_2\;\|Q_2\|_2^2
        \spaceeq
        \|N-{M_0}\|_2\\
        &\leq\;\;
        \|N-{M_0}\|_{\operatorname{F}}
        \;\;<\;\;
        \frac{1}{n^2}.
    \end{align*}
    In the expression above, the first inequality comes from the definition of the spectral norm and the Cauchy--Schwarz inequality, the second inequality is due to the submultiplicativity of the spectral norm, the first equality follows from the fact that $\|Q_2\|_2=1$ since the columns of $Q_2$ are orthonormal, the third inequality is a standard property of matrix norms (see~\cite[Thm.~5.6.34]{Horn2012matrix}), and the fourth inequality holds since $N\in\mathscr{B}_{\mathscr{U}^{\operatorname{a}}}({M_0};{\frac{1}{n^2}})$.
    It follows that
    \begin{align*}
        \bv^T\widehat{N}\bv
        &\spaceeq
        \bv^T\widehat{{M_0}}\bv+\bv^T\widehat{N-{M_0}}\bv
        \;\;\geq\;\;
        \bv^T\widehat{{M_0}}\bv-\lvert\bv^T\widehat{N-{M_0}}\bv\rvert 
        \;\;>\;\;
        \frac{1}{n}-\frac{1}{n^2}
        \;\;\geq\;\;
        0,
    \end{align*}
    thus proving the claim.

    We are left to show that $N\in\mathscr{U}$. It suffices to prove that $N$ is entrywise nonnegative, as all other conditions describing $\mathscr{U}$ are implied by the fact that $N\in\mathscr{U}^{\operatorname{a}}$. Take $x,y\in\Vset(\X)$ and $a,b\in\Vset(\A)$. If $x=y$ and $a\neq b$, $(\be_x\otimes\be_a)^TN(\be_y\otimes\be_b)=0$ by (\ref{r1}). Otherwise, noticing that $\|M_0-N\|^2_{\operatorname{F}}$
    is the sum of the squares of the entries of $M_0-N$, we find
\begin{align*}
    \lvert(\be_x\otimes\be_a)^T(M_0-N)(\be_y\otimes\be_b)\rvert
    \;\;\leq\;\;
    \|M_0-N\|_{\operatorname{F}}
    \;\;<\;\;
    \frac{1}{n^2}.
\end{align*} 
Noting that $(\be_x\otimes\be_a)^TM_0(\be_y\otimes\be_b)\geq \frac{1}{n^2}$, it follows that  $(\be_x\otimes\be_a)^TN(\be_y\otimes\be_b)> 0$, which establishes that $N\geq 0$ and concludes the proof of the proposition.
\end{proof}

\noindent Finally, we prove Proposition~\ref{prop_no_Slater_condition_for_SDP_bis}, implying that there exist no Slater points for the program~\eqref{eqn_SDP_W}.
\begin{proof}[Proof of Proposition~\ref{prop_no_Slater_condition_for_SDP_bis}]
    \noindent Given $M_0\in \mathscr{W}\cap\mathscr P$ and ${r}>0$, choose two distinct vertices $x,y\in\Vset(\X)$ and a vertex $a\in\Vset(\A)$. Consider the matrix $H=\frac{1}{n}(J_p-\be_x\be_y^T-\be_y\be_x^T)\otimes I_n+(\be_x\be_y^T+\be_y\be_x^T)\otimes\be_a\be_a^T$ and the number $s=\min(\frac{r}{2p^2+1},1)$, and define $N=(1-s)M_0+sH$. It is straightforward to check that $H\in\mathscr{W}$; since $\mathscr{W}$ is a convex set, it follows that $N\in\mathscr{W}\subseteq\mathscr{W}^{\operatorname{a}}$. Observe that each matrix $M\in\mathscr{W}$ satisfies $\|M\|_{\operatorname{F}}\leq p^2$
    (because of (\ref{r6}) and the entrywise nonnegativity of the entries). Therefore,
\begin{align*}
    \|N-M_0\|_{\operatorname{F}}
    \spaceeq
    \|s(H-M_0)\|_{\operatorname{F}}
    \spaceeq
    s\|H-M_0\|_{\operatorname{F}}
    \;\;\leq\;\;
    s(\|H\|_{\operatorname{F}}+\|M_0\|_{\operatorname{F}})
    \;\;\leq\;\;
    2p^2s
    \;\;<\;\;
    r,
\end{align*}
thus showing that $N\in\mathscr{B}_{\mathscr{W}^{\operatorname{a}}}(M_0;{r})$. We now prove that $N\not\in\mathscr{P}$. Let the space $Z$ and the matrices $Q_1,Q_2,Q$
be defined in the same way as in the proof of Proposition~\ref{prop_Slater_condition_for_SDP}. Since $\mathscr{U}\cap\mathscr P= \mathscr{W}\cap\mathscr P$, we have that $M_0\in\mathscr{U}$ and, thus, $Q^TM_0Q=\begin{bmatrix}
    O&O\\O&Q_2^TM_0Q_2
\end{bmatrix}$. Noting that $H\in\mathscr{W}$ and, hence, it satisfies (\ref{r6}), we see that $[(\be_{x'}-\be_{y'})\otimes\bone_n]^TH[(\be_{x''}-\be_{y''})\otimes\bone_n]=0$ for each $x',x'',y',y''\in\Vset(\X)$. Hence, $\bw^TH\bw'=0$ for each $\bw,\bw'\in Z$, which implies that $Q_1^THQ_1=O$.
We deduce that
\begin{align}
\label{eqn_2806_100}
    Q^TNQ
    \spaceeq
    (1-s)Q^TM_0Q+sQ^THQ
    \spaceeq
    \begin{bmatrix}
    O&sQ_1^THQ_2\\sQ_2^THQ_1&Q_2^TNQ_2
\end{bmatrix}.
\end{align}
Letting $\bz=(\be_x-\be_y)\otimes\bone_n$,
observe that
\begin{align*}
    H\bz
    &\spaceeq
    \frac{1}{n}(J_p-\be_x\be_y^T-\be_y\be_x^T)(\be_x-\be_y)\otimes(I_n\bone_n)+(\be_x\be_y^T+\be_y\be_x^T)(\be_x-\be_y)\otimes(\be_a\be_a^T\bone_n)\\
    &\spaceeq
    \frac{1}{n}(\be_x-\be_y)\otimes\bone_n+(\be_y-\be_x)\otimes\be_a,
\end{align*}
which does not belong to $Z$ (as we are assuming $x\neq y$ and $n\geq 2$). It follows that $Q_2^TH\bz\neq\bzero$; indeed, otherwise, we would have $H\bz\in(Z^\perp)^\perp=Z$. Since $\bz\in Z$, we have $\bz=Q_1\bv$ for some vector $\bv$. We deduce that $Q_2^THQ_1\bv\neq\bzero$, which means that $Q_2^THQ_1\neq O$. It is well known that, if a diagonal entry of a positive semidefinite matrix is zero, the corresponding row and column are zero (see~\cite[Obs.~7.1.10]{Horn2012matrix}). Looking at~\eqref{eqn_2806_100}, we deduce that $Q^TNQ\not\in\mathscr{P}$, thus yielding $N\not\in\mathscr{P}$, as needed.
\end{proof}

\subsection{SDP${}^{\epsilon}$ vs. BA${}^k$}
\label{subsec_SDPeps_vs_BAk}
The goal of this subsection is to establish the following result, which states that the test $\SDP^\epsilon$ is not dominated by the $\BA$ hierarchy for $\epsilon$ sufficiently small.
\begin{thm}
\label{thm_no_dominance_SDPeps_BAk}
For each $k\in\N$ there exists $\epsilon>0$ such that $\SDP^{\epsilon}\not\preceq \BA^k$.
\end{thm}
Since $\SDP^\epsilon\preceq\SDP\preceq\SDA$, it will follow that $\SDP\not\preceq\BA^k$ and $\SDA\not\preceq\BA^k$.\footnote{It is, however, an open question how the algorithms compare in terms of \emph{solvability} of PCSPs, see~Remark~\ref{rem_solvability_preorder}.}
In order to prove Theorem~\ref{thm_no_dominance_SDPeps_BAk}, we need to exhibit a class of digraphs that are correctly classified by $\SDP^\epsilon$ and not by $\BA^k$.
It turns out that cliques can serve as separating instances. Indeed, a result from~\cite{cz23stoc:ba} implies that the $\BA$ hierarchy is not sound on cliques. In contrast, $\SDP^\epsilon$ is able to correctly classify cliques provided that $\epsilon$ is small enough---as we shall prove next by leveraging the framework of association schemes.

We start by adapting the machinery developed in the previous sections to $\SDP^{\epsilon}$, thus obtaining the following result (akin to Corollary~\ref{cor_critierion_acceptance_SDP}).

\begin{prop}
\label{prop_critierion_acceptance_SDP_eps}
Let $\X$ and $\A$ be generously transitive digraphs, let $P$ and $\tilde{P}$ be the character tables of $\X$ and $\A$, respectively, let $\epsilon>0$, and suppose that $\SDP^\epsilon(\X,\A)=\YES$. Then there exists a real entrywise-nonnegative $\lvert\Orb(\X)\rvert\times\lvert\Orb(\A)\rvert$ matrix $V$ such that

    \begin{multicols}{2}
    \begin{itemize}
        \item[$(c_1)$\labeltext{$c_1$}{c1NEW}] $P V \tilde P^T\geq 0$;
        \item[$(c_2)$\labeltext{$c_2$}{c2NEW}] $V\bmu^\A=\bone$;
    \end{itemize}
    \end{multicols}
\vspace{-.5cm}
    \begin{itemize}
        \item[$(c_3)$\labeltext{$c_3$}{c3NEW}] $v_{\omega\tilde\omega}=0$ if $\omega$ is the diagonal orbital of $\X$ and $\tilde\omega$ is a non-diagonal orbital of $\A$;
        \item[$(c_4')$\labeltext{$c_4'$}{c4NEW}] $v_{\omega\tilde\omega}\leq\epsilon$
        if $\omega$ is an edge orbital of $\X$ and $\tilde\omega$ is a non-edge orbital of $\A$.
    \end{itemize} 
\end{prop}
\begin{proof}
    Let $M\in\mathscr{U}\cap\mathscr{P}$ be a matrix witnessing that $\SDP^\epsilon(\X,\A)=\YES$, and observe that $f(M)\leq\epsilon$ by the definition of $\SDP^\epsilon$ (where $f(M)$ is defined in~\eqref{eqn_f_of_M}). Given two automorphisms $\xi\in\Aut(\X)$ and $\alpha\in\Aut(\A)$, consider the matrix
\begin{align*}
        M^{(\xi,\alpha)}
        \spaceeq
        (Q_\xi\otimes Q_\alpha^T)M(Q_\xi^T\otimes Q_\alpha).
    \end{align*}
    The same argument as in the proof of Proposition~\ref{prop_relaxation_matrix_preservation_homo} (see Appendix~\ref{sec:app}) shows that $M^{(\xi,\alpha)}\in\mathscr{U}\cap\mathscr{P}$ and $f(M^{(\xi,\alpha)})=f(M)$.
    Hence, a minor modification of the proof of Proposition~\ref{relaxation_matrix_can_be_balanced} implies that the matrix
    \begin{align*}
    \overline M
    \spaceeq
    \frac{1}{\lvert\Aut(\X)\rvert\cdot\lvert\Aut(\A)\rvert}\sum_{\substack{\xi\in\Aut(\X)\\\alpha\in\Aut(\A)}}M^{(\xi,\alpha)}
\end{align*}
    is a balanced matrix for $\X,\A$ and satisfies $\overline M\in \mathscr{U}\cap\mathscr{P}$ and $f(\overline M)=f(M)$. In particular, this means that $(\be_x\otimes\be_a)^T\overline M(\be_y\otimes\be_b)\leq\epsilon$ for each $(x,y)\in\Eset(\X)$, $(a,b)\in \Vset(\A)^2\setminus\Eset(\A)$. We can then express $\overline M$ in the basis $\mathscr{R}$ as per Proposition~\ref{prop_orbital_matrix_of_balanced_matrix}, and conclude the proof by making use of Theorem~\ref{prop_decomposing_spectrum_M_orbital} in the same way as in the proof of Corollary~\ref{cor_critierion_acceptance_SDP}.
\end{proof}

We point out that the necessary condition for $\SDP^\epsilon$ acceptance expressed in Proposition~\ref{prop_critierion_acceptance_SDP_eps} is not sufficient, unlike the similar condition for $\SDP$ in Corollary~\ref{cor_critierion_acceptance_SDP}. 
    The reason is that the existence of a feasible solution to~\eqref{eqn_SDP_U} with objective-function value at most $\epsilon$ does not guarantee that $\SDP^{\epsilon}(\X,\A)=\YES$. Indeed, the only piece of information we would be able to derive in this case is that the optimal value $\nu$ of the program~\eqref{eqn_SDP_U} satisfies $\nu\leq\epsilon$. Theorem~\ref{thm_ellipsoid_method_is_fast_with_Slater} would then guarantee that the ellipsoid method finds a feasible solution $M^\ast$ with $f(M^\ast)\leq\epsilon+\nu\leq 2\epsilon$, while to ensure that $\SDP^\epsilon(\X,\A)=\YES$ we need that $f(M^\ast)\leq\epsilon$.

We now show that $\SDP^{\epsilon}$ is sound on cliques for $\epsilon$ small enough, by using the description of the corresponding association schemes.
Of course, it follows that the same holds for $\SDP$ and $\SDA$, since $\SDP^{\epsilon}\preceq\SDP\preceq\SDA$.

\begin{prop}
\label{prop_no_acceptance_cliques_SDP_epsilon}
Let $p,n\geq 2$ and $0<\epsilon<\frac{1}{n^3}$. Then
  $\SDP^{\epsilon}(\K_p,\K_n)=\YES$ if and only if 
  $p\leq n$.
\end{prop}
\begin{proof}
    If $p\leq n$, $\K_p\to\K_n$; since $\SDP^{\epsilon}$ is a complete test by Theorem~\ref{thm_SDP_eps_complete_pol}, it follows that $\SDP^{\epsilon}(\K_p,\K_n)=\YES$.
    Conversely, suppose that $\SDP^{\epsilon}(\K_p,\K_n)=\YES$. 
    Recall the description of the association schemes for cliques in Remark~\ref{rem_association_scheme_cliques}, and observe that it implies $\bmu^{\K_n}=\begin{bmatrix}
        n\\n^2-n
    \end{bmatrix}$.
    Let $V=\begin{bmatrix}
        a&b\\c&d
    \end{bmatrix}$ be the $2\times 2$ matrix whose existence is guaranteed by Proposition~\ref{prop_critierion_acceptance_SDP_eps}, and notice that the requirements (\ref{c3NEW}) and (\ref{c4NEW}) yield $b=0$ and $c\leq\epsilon$, respectively. As for (\ref{c2NEW}), it gives
    \begin{align*}
        \begin{bmatrix}
            1\\1
        \end{bmatrix}
        &\spaceeq
        \begin{bmatrix}
        a&0\\c&d
    \end{bmatrix}
    \begin{bmatrix}
        n\\n^2-n
    \end{bmatrix}
    \spaceeq
    \begin{bmatrix}
        an\\
        cn+d(n^2-n)
    \end{bmatrix},
    \end{align*}
    whence it follows that
        $a= c+d(n-1)=\frac{1}{n}$ and $c-d=\frac{n^2c-1}{n^2-n}$.
Letting $P$ and $\tilde P$ denote the character tables of $\K_p$ and $\K_n$, respectively, we obtain
\begin{align*}
    V\tilde P^T
    \spaceeq
    \begin{bmatrix}
        a&0\\c&d
    \end{bmatrix}
    \begin{bmatrix}
        1&1\\
        n-1&-1
    \end{bmatrix}
    \spaceeq
    \begin{bmatrix}
        a&a\\
        c+d(n-1)&c-d
    \end{bmatrix}
    \spaceeq
    \begin{bmatrix}
        \frac{1}{n}&\frac{1}{n}\\\frac{1}{n}&\frac{n^2c-1}{n^2-n}
    \end{bmatrix}.
\end{align*}
Therefore,
\begin{align*}
    PV\tilde P^T
    \spaceeq
    \begin{bmatrix}
        1&p-1\\1&-1
    \end{bmatrix}
    \begin{bmatrix}
        \frac{1}{n}&\frac{1}{n}\\\frac{1}{n}&\frac{n^2c-1}{n^2-n}
    \end{bmatrix}
    \spaceeq
    \begin{bmatrix}
        \frac{p}{n}&\frac{n-p+n^2c(p-1)}{n^2-n}
        \\
        0&\frac{1-nc}{n-1}
    \end{bmatrix}.
\end{align*}
The condition (\ref{c1NEW}) implies in particular that the $(1,2)$-th entry of the matrix above is nonnegative. Combining this with the fact that $c\leq\epsilon$ and the assumption that $\epsilon<\frac{1}{n^3}$ yields
\begin{align*}
    0
    \;\;\leq\;\;
    n-p+n^2c(p-1)
    \;\;<\;\;
    n-p+\frac{p-1}{n}
    \spaceeq
    \frac{n-1}{n}(n+1-p),
\end{align*}
which implies that $p\leq n$, as needed.
\end{proof}

\noindent On the other hand, the next result from~\cite{cz23stoc:ba}---used to rule out solvability of approximate graph colouring through the $\BA$ hierarchy---shows that $\BA^k(\X,\A)$ always accepts if $\A$ is a large-enough clique.

\begin{thm}[\cite{cz23stoc:ba}]
\label{them_acceptnace_hollow_shadow}
    Let $2\leq k\in\N$ and let $\X$ be a loopless digraph. Then $\BA^k(\X,\K_{k^2})=\YES$.
\end{thm}

\noindent It follows, in particular, that the $\BA$ hierarchy is not sound on cliques.
We can then finalise the proof of Theorem~\ref{thm_no_dominance_SDPeps_BAk}.

\begin{proof}[Proof of Theorem~\ref{thm_no_dominance_SDPeps_BAk}]
It is enough to prove the result for $k$ large enough since $\BA^k\preceq\BA^{k+1}$ for any $k\in\N$ (as follows from the definition of the $\BA$ hierarchy and as was shown in a more general setting of lift-and-project hierarchies of relaxations in~\cite{cz23soda:minions}). For $k\geq 2$, choose a positive $\epsilon<\frac{1}{k^6}$.
We have from Theorem~\ref{them_acceptnace_hollow_shadow} that $\BA^k(\K_{k^2+1},\K_{k^2})=\YES$, while  Proposition~\ref{prop_no_acceptance_cliques_SDP_epsilon} implies that $\SDP^\epsilon(\K_{k^2+1},\K_{k^2})=\NO$. We conclude that $\SDP^{\epsilon}\not\preceq \BA^k$, as required.
\end{proof}

\begin{rem} 
\label{rem_solvability_preorder}
The partial order ``$\preceq$'' described in Definition~\ref{defn_tests_and_annexes} concerns \emph{acceptance} of tests. A different way of comparing two tests is by considering \emph{solvability}. More precisely, let $\mathfrak{E}=\{(\A,\B)\in\mathfrak{D}^2:\A\to\B\}$. We say that a test $T$ \emph{solves} $(\A,\B)\in\mathfrak{E}$
if, for each $\X\in\mathfrak{D}$, $\X\to\A$ implies $T(\X,\A)=\YES$ and $\X\not\to\B$ implies $T(\X,\A)=\NO$. We can then define a preorder ``$\leq$'' on the set of tests by setting $T_1\leq T_2$ when, for any $(\A,\B)\in\mathfrak{E}$, if $T_1$ solves $(\A,\B)$ then $T_2$ also solves it.
It is not difficult to check that  ``$\leq$''  is finer than ``$\preceq$'' on complete tests; i.e., for two complete tests $T_1$ and $T_2$,
if $T_1\preceq T_2$ then $T_1\leq T_2$. Observe that ``$\leq$'' is not antisymmetric and, thus, it is not a partial order.

The results proved in the current subsection imply that $\SDA\not\preceq \BA^k$ for any $k\in\N$.
    On the other hand, the question of whether $\SDA\leq \BA^k$ for some $k$ is, to the best of the authors' knowledge, open, and Theorem~\ref{thm_no_dominance_SDPeps_BAk} does not provide any evidence supporting a negative answer.
    In fact, it is not even clear how $\SDP$ compares to the Sherali--Adams LP hierarchy---i.e., the lift-and-project hierarchy built on top of the $\BLP$ algorithm, see~\cite{tz17:sicomp,cz23soda:minions}---in terms of the preorder ``$\leq$''. Finally, we point out that,
    in the definition of ``$\leq$'', we may replace $\mathfrak{E}$ with the smaller set of pairs $(\A,\A)$ with $\A\in\mathfrak{D}$. The resulting preorder ``$\leq_{\mbox{\tiny CSP}}$'' would correspond to comparing solvability for CSPs rather than PCSPs.  
    It is known that $\SDP$ and any level of the Sherali--Adams LP hierarchy higher than $2$ are $\leq_{\mbox{\tiny CSP}}$-equivalent, as they all solve precisely those CSPs having bounded width~\cite{Barto16:sicomp,tz17:sicomp,tz18}.
\end{rem}

{\small
\bibliographystyle{plainurl}
\bibliography{cz_bibliography}
}

\appendix

\section{Appendix} 
\label{sec:app}
This section contains the proofs of 
the results stated in Section~\ref{sec_preliminaries}.

\begin{prop*}[Proposition~\ref{prop_equivalence_some_conditions_SDP} restated]
Let $M$ be a real $pn\times pn$ matrix. Then
\begin{itemize}
    \item[$(i)$] the conditions \emph{(\ref{r3})}, \emph{(\ref{r4})}, and \emph{(\ref{r5})} imply the condition \emph{(\ref{r6})};
    \item[$(ii)$] if $M\succcurlyeq 0$, the conditions \emph{(\ref{r3})}, \emph{(\ref{r4})}, and \emph{(\ref{r5})} are equivalent to the condition \emph{(\ref{r6})}.
\end{itemize}
\end{prop*}
\begin{proof}
    If $M$ satisfies (\ref{r3}) and (\ref{r4}), the expression on the left-hand side of (\ref{r6}) has the same value for each choice of $x,y\in \Vset(\X)$. Therefore, (\ref{r5}) yields
    \begin{align*}
        p^2
        &\spaceeq
        \bone_{pn}^TM\bone_{pn}
        \spaceeq
        (\bone_p\otimes\bone_n)^TM(\bone_p\otimes\bone_n)
        \spaceeq
        \sum_{x',y'\in \Vset(\X)}(\be_{x'}\otimes\bone_n)^TM(\be_{y'}\otimes\bone_n)\\
        &\spaceeq
        p^2 (\be_x\otimes\bone_n)^TM(\be_y\otimes\bone_n),
    \end{align*}
    which implies (\ref{r6}). 
    
    Suppose now that $M$ is positive semidefinite and satisfies (\ref{r6}), and consider a Cholesky decomposition $M=L^TL$. For each $x,y\in \Vset(\X)$, we find
    \begin{align*}
    0\spaceeq
        ((\be_x-\be_y)\otimes\bone_n)^T M ((\be_x-\be_y)\otimes\bone_n)
        \spaceeq
        \|L((\be_x-\be_y)\otimes\bone_n)\|^2.
    \end{align*}
We deduce that $L((\be_x-\be_y)\otimes\bone_n)=\bzero$, so 
\begin{align*}
    \bzero
    \spaceeq
    L^TL((\be_x-\be_y)\otimes\bone_n)
    \spaceeq
    M((\be_x-\be_y)\otimes\bone_n)
    \spaceeq
    M(\be_x\otimes\bone_n)-M(\be_y\otimes\bone_n),
\end{align*}
as needed to prove (\ref{r3}) (as well as (\ref{r4}), since $M$ is symmetric). Furthermore,
\begin{align*}
    \bone^T_{pn}M\bone_{pn}
    \spaceeq
    \sum_{(x,y)\in \Vset(\X)^2}(\be_x\otimes\bone_n)^T M (\be_y\otimes\bone_n)
    \spaceeq
    \sum_{(x,y)\in \Vset(\X)^2} 1
    \spaceeq
    p^2,
\end{align*}
which shows that (\ref{r5}) holds, as well.
\end{proof}

\begin{prop*}[Proposition~\ref{prop_matrix_formulation_SDP_SDA} restated]
    Let $\X,\A$ be digraphs. Then
    \begin{itemize}
        \item[$(i)$] $\SDP(\X,\A)=\YES$ if and only if there exists an $\SDP$-matrix for $\X,\A$;
        \item[$(ii)$] if $\X$ is loopless, $\SDA(\X,\A)=\YES$ if and only if there exist an $\SDP$-matrix $M$ and an $\AIP$-matrix $N$ for $\X,\A$ such that $N\circ ((I_p+\adj(\X))\otimes J_n)\;\;\triangleleft\;\;M$.
    \end{itemize}
\end{prop*}
\begin{proof}
    To prove the first statement of the proposition, let the vectors $\{\blambda_{x,a}:x\in \Vset(\X),a\in \Vset(\A)\}$ witness that $\SDP(\X,\A)=\YES$, where each $\blambda_{x,a}$ belongs to $\R^{pn}$. Consider the $pn\times pn$ matrix $L$ whose columns are the vectors $\blambda_{x,a}$, and let $M=L^TL$. We claim that $M$ is an $\SDP$-matrix for $\X,\A$. 
Since $M\succcurlyeq 0$ by construction and $M\geq 0$ by (\ref{SDP1}), it is enough to show that $M$ is a relaxation matrix.   
Observe that
\begin{align}
\label{eqn_1725_2205}
\notag
    (\be_x\otimes\be_a)^TM(\be_y\otimes\be_b)
    &\spaceeq
    (\be_x\otimes\be_a)^TL^TL(\be_y\otimes\be_b)
    \spaceeq
    (L(\be_x\otimes\be_a))^T(L(\be_y\otimes\be_b))\\
    &\spaceeq
    \blambda_{x,a}\cdot\blambda_{y,b}
\end{align}
for each $x,y\in \Vset(\X)$, $a,b\in \Vset(\A)$. We obtain
\begin{align*}
    \be_a^T(\be_x\otimes I_n)^T M (\be_x\otimes I_n)\be_b
    &\spaceeq
    (1\otimes\be_a)^T(\be_x\otimes I_n)^T M (\be_x\otimes I_n)(1\otimes\be_b)\\
    &\spaceeq
    (\be_x\otimes\be_a)^TM(\be_x\otimes\be_b)
    \spaceeq
    \blambda_{x,a}\cdot\blambda_{x,b},
\end{align*}
so the condition (\ref{r1}) is guaranteed by (\ref{SDP2}). Also, (\ref{r2}) directly follows from (\ref{SDP3}) using~\eqref{eqn_1725_2205}.
From (\ref{SDP4}), we see that 
\begin{align}
\label{eqn_3105_1223}
\notag
    (\be_x\otimes\bone_n)^T M (\be_y\otimes\bone_n)
    &\spaceeq
    \Big(\be_x\otimes\sum_{a\in \Vset(\A)}\be_a\Big)^T M \Big(\be_y\otimes\sum_{b\in \Vset(\A)}\be_b\Big)\\
    \notag
    &\spaceeq
    \sum_{a,b\in \Vset(\A)}(\be_x\otimes\be_a)^TM(\be_y\otimes\be_b)
    \spaceeq
    \sum_{a,b\in \Vset(\A)}\blambda_{x,a}\cdot\blambda_{y,b}\\
    &\spaceeq
    \blambda_{x,\A}\cdot\blambda_{y,\A}
    \spaceeq
    1
\end{align}
for each $x,y\in \Vset(\X)$. It follows from Proposition~\ref{prop_equivalence_some_conditions_SDP} that the conditions (\ref{r3}), (\ref{r4}), and (\ref{r5}) are satisfied.

Conversely, let $M$ be an $\SDP$-matrix for $\X,\A$. Since $M\succcurlyeq 0$, we can find a Cholesky decomposition $M=L^TL$ for some $pn\times pn$ matrix $L$. Letting $\blambda_{x,a}=L(\be_x\otimes\be_a)$ for each $x\in \Vset(\X)$, $a\in \Vset(\A)$, we see that the equation~\eqref{eqn_1725_2205} holds. Thus, the conditions (\ref{SDP1}), (\ref{SDP2}), and (\ref{SDP3}) follow from the entrywise nonnegativity of $M$, (\ref{r1}), and (\ref{r2}), respectively. Moreover, since $M$ satisfies (\ref{r6}) by Proposition~\ref{prop_equivalence_some_conditions_SDP}, we can reverse~\eqref{eqn_3105_1223} to see that $\blambda_{x,\A}\cdot\blambda_{y,\A}=(\be_x\otimes\bone_n)^T M (\be_y\otimes\bone_n)=1$ for each $x,y\in \Vset(\X)$,
which gives (\ref{SDP4}).  

Next, we establish the second statement of the proposition. Let $\blambda$ and $\mu$ witness that $\SDA(\X,\A)=\YES$, and let $M$ be the $\SDP$-matrix corresponding to $\blambda$ as defined in part $(i)$ of this proposition. Consider the $pn\times pn$ matrix $N$ defined by
    \begin{align}
    \label{eqn_0905_1102}
    (\be_x\otimes\be_a)^TN(\be_y\otimes\be_b)
        &\spaceeq
        \left\{
        \begin{array}{llll}
            \mu_{x,a} & \mbox{ if } x=y,\, a=b \\
             \mu_{x,a}\mu_{y,b}& \mbox{ if } x\neq y,\, (x,y)\not\in\Eset(\X)\\
             \mu_{(x,y),(a,b)}& \mbox{ if }(x,y)\in\Eset(\X),\,(a,b)\in\Eset(\A)\\
             0& \mbox{ otherwise}.
        \end{array}
        \right.
    \end{align}
    We claim that $N$ is a relaxation matrix for $\X,\A$. Since all entries of $N$ are integral, this would mean that $N$ is an $\AIP$-matrix. The conditions (\ref{r1}) and (\ref{r2}) are immediate from~\eqref{eqn_0905_1102}, using the assumption that $\X$ is loopless.
From (\ref{AIP1}) and (\ref{AIP2}), we straightforwardly check that the equations
    \begin{align}
    \label{eqn_0905_1103}
    (\be_x\otimes\be_a)^T N(\be_y\otimes\bone_n)
        \spaceeq
        (\be_x\otimes\be_a)^T N^T(\be_y\otimes\bone_n)
        \spaceeq
        \mu_{x,a}
    \end{align}
    hold for each $x,y\in \Vset(\X)$ and $a\in \Vset(\A)$. In particular, this means that the vectors $N(\be_y\otimes\bone_n)$ and $N^T(\be_y\otimes\bone_n)$ do not depend on $y$, which yields (\ref{r3}) and (\ref{r4}). Finally, using~\eqref{eqn_0905_1103} and (\ref{AIP1}), we find
    \begin{align*}
        \bone_{pn}^TN\bone_{pn}
        \spaceeq
        \sum_{\substack{(x,y)\in \Vset(\X)^2\\a\in \Vset(\A)}}(\be_x\otimes\be_a)^TN(\be_{y}\otimes\bone_n)
        \spaceeq
        \sum_{\substack{(x,y)\in \Vset(\X)^2\\a\in \Vset(\A)}}\mu_{x,a}
        \spaceeq
        \sum_{(x,y)\in \Vset(\X)^2}1
        \spaceeq
        p^2,
    \end{align*}
    which yields (\ref{r5}) and proves the claim. Since $N\circ ((I_p+\adj(\X))\otimes J_n)=N\circ(I_p\otimes J_n)+N\circ(\adj(\X)\otimes J_n)$,
    \begin{align*}
        \supp(N\circ ((I_p+\adj(\X))\otimes J_n))
        &\;\;\subseteq\;\;
        \supp(N\circ (I_p\otimes J_n))\;\cup\;
        \supp(N\circ (\adj(\X)\otimes J_n))\\
        &\spaceeq
        \{((x,a),(x,a)):\mu_{x,a}\neq 0\}\\
        &\;\;\cup\;\;\{((x,a),(y,b)):(x,y)\in\Eset(\X),\,(a,b)\in\Eset(\A),\,\mu_{(x,y),(a,b)}\neq 0\}.
    \end{align*}
By~\eqref{eqn_refinement_condition_SDA}, $\mu_{x,a}\neq 0$ implies $\|\blambda_{x,a}\|\neq 0$, which means that $(\be_x\otimes\be_a)^TM(\be_x\otimes\be_a)\neq 0$. Similarly, $\mu_{(x,y),(a,b)}\neq 0$ implies that $\blambda_{x,a}\cdot\blambda_{y,b}\neq 0$, which means that $(\be_x\otimes\be_a)^TM(\be_y\otimes\be_b)\neq 0$. It follows that $N\circ ((I_p+\adj(\X))\otimes J_n)\;\triangleleft\;M$, as required.

Conversely, let $M$ and $N$ be an $\SDP$-matrix and an $\AIP$-matrix for $\X,\A$, respectively. As was shown above, the vectors $\blambda_{x,a}$ obtained through a Cholesky decomposition of $M$ give a solution to $\SDP(\X,\A)$. Moreover, we define $\mu_{x,a}=(\be_x\otimes\be_a)^TN(\be_x\otimes\be_a)$ for $x\in \Vset(\X)$ and $a\in \Vset(\A)$, and $\mu_{(x,y),(a,b)}=(\be_x\otimes\be_a)^TN(\be_y\otimes\be_b)$ for $(x,y)\in\Eset(\X)$ and $(a,b)\in\Eset(\A)$. Reversing the argument above, it easily follows from the conditions defining an $\AIP$-matrix that $\mu$ satisfies (\ref{AIP1}) and (\ref{AIP2}), while the requirement~\eqref{eqn_refinement_condition_SDA} follows from the fact that $N\circ ((I_p+\adj(\X))\otimes J_n)\;\triangleleft\;M$.
\end{proof}
Recall that, given two finite sets $R$ and $S$ and a function $f:R\to S$, $Q_f$ denotes the $\lvert R\rvert\times \lvert S\rvert$ matrix whose $(r,s)$-th entry is $1$ if $f(r)=s$, $0$ otherwise.

\begin{prop*}
[Proposition~\ref{prop_relaxation_matrix_preservation_homo} restated]
    Let $\X,\X',\A,\A'$ be digraphs, let $f:\X'\to\X$ and $g:\A\to\A'$ be homomorphisms, and let $M$ be a relaxation matrix for $\X,\A$. Then
\begin{align*}
    M^{(f,g)}\spaceeq(Q_f\otimes Q_g^T)M(Q_f^T\otimes Q_g)
\end{align*}
is a relaxation matrix for $\X',\A'$. Furthermore, if $M$ is an $\SDP$-matrix (resp. $\AIP$-matrix) for $\X,\A$, then $M^{(f,g)}$ is an $\SDP$-matrix (resp. $\AIP$-matrix) for $\X',\A'$.
\end{prop*}
\begin{proof}
    For $x',y'\in \Vset(\X')$ and $a',b'\in \Vset(\A')$, we find 
    \begin{align}
    \label{eqn_0303_1848}
    \notag
        (\be_{x'}\otimes\be_{a'})^TM^{(f,g)}(\be_{y'}\otimes\be_{b'})
        &\spaceeq
        (\be_{x'}\otimes\be_{a'})^T(Q_f\otimes Q_g^T)M(Q_f^T\otimes Q_g)(\be_{y'}\otimes\be_{b'})\\
        \notag
        &\spaceeq
        (\be_{x'}^TQ_f\otimes \be_{a'}^TQ_g^T)M(Q_f^T\be_{y'}\otimes Q_g\be_{b'})\\
        &\spaceeq
        \sum_{\substack{a\in g^{-1}(a')\\b\in g^{-1}(b')}}(\be_{f(x')}\otimes\be_{a})^TM(\be_{f(y')}\otimes \be_{b}).
    \end{align}
Note that, if $x'=y'$ and $a'\neq b'$, we have $f(x')=f(y')$ and $g^{-1}(a')\cap g^{-1}(b')=\emptyset$. Similarly, if $(x',y')\in\Eset(\X')$ and $(a',b')\not\in \Eset(\A')$, we have $(f(x'),f(y'))\in\Eset(\X)$ and $(g^{-1}(a')\times g^{-1}(b'))\,\cap\, \Eset(\A)=\emptyset$, since $f$ and $g$ are homomorphisms. Hence, the conditions (\ref{r1}) and (\ref{r2}) for $M^{(f,g)}$ follow from the same conditions applied to $M$. 
Let $p'=\lvert\Vset(\X')\rvert$ and $n'=\lvert\Vset(\A')\rvert$. Observe that
    \begin{align*}
        M^{(f,g)}(\be_{x'}\otimes\bone_{n'})
        &\spaceeq
        (Q_f\otimes Q_g^T)M(Q_f^T\otimes Q_g)(\be_{x'}\otimes\bone_{n'})
        \spaceeq
        (Q_f\otimes Q_g^T)M(\be_{f(x')}\otimes\bone_n).
    \end{align*}
    Therefore, (\ref{r3}) follows from the same condition applied to $M$, and (\ref{r4}) follows analogously. Fix $\tilde x,\tilde y\in \Vset(\X')$. From (\ref{r3}) and (\ref{r4}), we find
    \begin{align*}
        %&
        \bone_{pn}^TM\bone_{pn}
        &\spaceeq
        (\bone_p\otimes\bone_n)^TM(\bone_p\otimes\bone_n)
        \spaceeq
        \sum_{x,y\in \Vset(\X)}(\be_x\otimes\bone_n)^TM(\be_y\otimes\bone_n)\\
        &\spaceeq
        p^2(\be_{f(\tilde x)}\otimes\bone_n)^TM(\be_{f(\tilde y)}\otimes\bone_n)
        \intertext{and, similarly,}
\bone_{p'n'}^TM^{(f,g)}\bone_{p'n'}
        &\spaceeq
        (p')^2(\be_{\tilde x}\otimes\bone_{n'})^TM^{(f,g)}(\be_{\tilde y}\otimes\bone_{n'}).
    \end{align*}
    Using~\eqref{eqn_0303_1848}, we obtain
    \begin{align*}
        (\be_{\tilde x}\otimes\bone_{n'})^TM^{(f,g)}(\be_{\tilde y}\otimes\bone_{n'})
        &\spaceeq
        \sum_{a',b'\in \Vset(\A')}(\be_{\tilde x}\otimes\be_{a'})^TM^{(f,g)}(\be_{\tilde y}\otimes\be_{b'})\\
        &\spaceeq
        \sum_{a',b'\in \Vset(\A')}\sum_{\substack{a\in g^{-1}(a')\\b\in g^{-1}(b')}}(\be_{f(\tilde x)}\otimes\be_{a})^TM(\be_{f(\tilde y)}\otimes \be_{b})\\
        &\spaceeq
        \sum_{a,b\in \Vset(\A)}(\be_{f(\tilde x)}\otimes\be_{a})^TM(\be_{f(\tilde y)}\otimes \be_{b})\\
        &\spaceeq
        (\be_{f(\tilde x)}\otimes\bone_n)^TM(\be_{f(\tilde y)}\otimes \bone_n).
    \end{align*}
    From (\ref{r5}) applied to $M$, we get
    \begin{align*}
        \bone_{p'n'}^TM^{(f,g)}\bone_{p'n'}
        \spaceeq
        (p')^2(\be_{f(\tilde x)}\otimes\bone_n)^TM(\be_{f(\tilde y)}\otimes \bone_n)
        \spaceeq
        (p')^2\frac{\bone_{pn}^TM\bone_{pn}}{p^2}
        \spaceeq
        (p')^2,
    \end{align*}
    which shows that $M^{(f,g)}$ satisfies (\ref{r5}) and, thus, it is a relaxation matrix.

    To prove the last part of the statement, we simply observe that  $Q_f$ and $Q_g$ are Boolean matrices, so the product in the right-hand side of~\eqref{eqn_0106_1611} preserves the nonnegativity and the integrality of $M$. Moreover, since we can write $M^{(f,g)}=PMP^T$ for $P=Q_f\otimes Q_g^T$, it easily follows that  $M^{(f,g)}\succcurlyeq 0$ when $M\succcurlyeq 0$.
\end{proof}

\noindent To prove Proposition~\ref{prop_monotonicity_SDP_SDA_wrt_homomorphisms} we will need the following lemma, whose trivial proof is omitted.

\begin{lem}
\label{lem_basic_support}
    Let $A,B$ be $q\times r$ matrices and let $C,D$ be $r\times s$ matrices, and suppose that $A\triangleleft B$ and $C\triangleleft D$ and that $B,D\geq 0$. Then, $AC\;\triangleleft\; BD$.
\end{lem}

\begin{prop*}
[Proposition~\ref{prop_monotonicity_SDP_SDA_wrt_homomorphisms} restated]
    Let $\X,\X',\A,\A'$ be digraphs such that $\X'\to\X$ and $\A\to\A'$. Then
    \begin{itemize}
        \item[$(i)$] $\SDP(\X,\A)=\YES$ implies $\SDP(\X',\A')=\YES$;
        \item[$(ii)$] if $\X$ is loopless, $\SDA(\X,\A)=\YES$ implies $\SDA(\X',\A')=\YES$.
    \end{itemize}
\end{prop*}
\begin{proof}
    Part $(i)$ directly follows from Proposition~\ref{prop_matrix_formulation_SDP_SDA} and Proposition~\ref{prop_relaxation_matrix_preservation_homo}. To prove part $(ii)$, let $M$ and $N$ be the $\SDP$-matrix and the $\AIP$-matrix for $\X,\A$ whose existence is guaranteed by Proposition~\ref{prop_matrix_formulation_SDP_SDA}. Choose two homomorphisms $f:\X'\to\X$ and $g:\A\to\A'$.
Proposition~\ref{prop_relaxation_matrix_preservation_homo} gives that $M^{(f,g)}$ and $N^{(f,g)}$ are an $\SDP$-matrix and an $\AIP$-matrix for $\X',\A'$, respectively. Let $N_{\operatorname{d}}=N\circ(I_p\otimes J_n)$, $N_{\operatorname{e}}=N\circ(\adj(\X)\otimes J_n)$, $N^{(f,g)}_{\operatorname{d}}=N^{(f,g)}\circ(I_{p'}\otimes J_{n'})$, and $N^{(f,g)}_{\operatorname{e}}=N^{(f,g)}\circ(\adj(\X')\otimes J_{n'})$. Note that, if $\X'\to\X$ and $\X$ is loopless, then $\X'$ is also loopless. Hence,
if we establish that
    \begin{align}
        \label{eqn_0206_1208}N^{(f,g)}_{\operatorname{d}}+
        N^{(f,g)}_{\operatorname{e}}\;\;\triangleleft\;\;M^{(f,g)},
    \end{align}
    we can conclude that $\SDA(\X',\A')=\YES$ using Proposition~\ref{prop_matrix_formulation_SDP_SDA}. 
    Call $P=Q_f\otimes Q_g^T$, so that $N^{(f,g)}=PNP^T$ and $M^{(f,g)}=PMP^T$. A straightforward computation shows that $(PN_{\operatorname{d}}P^T)\circ(I_{p'}\otimes J_{n'})=N^{(f,g)}_{\operatorname{d}}$ and, similarly, $(PN_{\operatorname{e}}P^T)\circ(\adj(\X')\otimes J_{n'})=N^{(f,g)}_{\operatorname{e}}$. In particular, this means that $N^{(f,g)}_{\operatorname{d}}\;\triangleleft\; PN_{\operatorname{d}}P^T$ and $N^{(f,g)}_{\operatorname{e}}\;\triangleleft \;PN_{\operatorname{e}}P^T$. Since $\X'$ is loopless, the support of $PN_{\operatorname{d}}P^T$ and the support of $PN_{\operatorname{e}}P^T$ are disjoint, which means that the support of their sum equals the union of their supports. Hence, 
    \begin{align*}
        \supp(N^{(f,g)}_{\operatorname{d}}+N^{(f,g)}_{\operatorname{e}})
        &\;\;\subseteq\;\;
        \supp(N^{(f,g)}_{\operatorname{d}})\cup\supp(N^{(f,g)}_{\operatorname{e}})
        \;\;\subseteq\;\;
        \supp(PN_{\operatorname{d}}P^T)\cup\supp(PN_{\operatorname{e}}P^T)\\
        &\spaceeq
        \supp(PN_{\operatorname{d}}P^T+PN_{\operatorname{e}}P^T).
    \end{align*}
Recall that, by Proposition~\ref{prop_matrix_formulation_SDP_SDA}, $N_{\operatorname{d}}+N_{\operatorname{e}}\;\triangleleft\; M$. Since $M$ and $P$ are entrywise nonnegative, we can apply Lemma~\ref{lem_basic_support} to obtain
\begin{align*}
    N^{(f,g)}_{\operatorname{d}}+N^{(f,g)}_{\operatorname{e}}
    \;\;\triangleleft\;\;
    PN_{\operatorname{d}}P^T+PN_{\operatorname{e}}P^T
    \spaceeq
    P(N_{\operatorname{d}}+N_{\operatorname{e}})P^T
    \;\;\triangleleft\;\;
    PMP^T
    \spaceeq
    M^{(f,g)},
\end{align*}
    thus establishing~\eqref{eqn_0206_1208}.
\end{proof}

\end{document}